\documentclass[twocolumn]{autart}    % Enable this line andL disable the 
\usepackage{natbib}
\usepackage{graphicx}          % Include this line if your 
\usepackage[utf8]{inputenc}  % solo se usi pdflatex
\usepackage[T1]{fontenc}
\usepackage{textcomp}
\usepackage{comment}

\usepackage{subfigure}
 
\usepackage{wrapfig}

\usepackage{xcolor}
\usepackage{xparse}
\usepackage{mathrsfs}
\usepackage{graphics} % for pdf,  bitmapped graphics files
\usepackage{amsmath} % assumes amsmath package installed
\usepackage{amssymb}  % assumes amsmath package installed

\usepackage{algorithm}
\floatname{algorithm}{\bf{Algorithm}}

\newenvironment{proof}[1][Proof]{\smallskip\par\noindent\textbf{#1.} }{\hfill$\square$\par \smallskip}

\usepackage{enumitem} 
\usepackage{hyperref}%
\hypersetup{colorlinks=true,  linkcolor=blue,  breaklinks=true,  urlcolor=blue,  citecolor=blue}

\usepackage{lipsum}
\usepackage{mwe}
\usepackage{soul}
\usepackage{url}
\DeclareSymbolFont{bbold}{U}{bbold}{m}{n}
\DeclareSymbolFontAlphabet{\mathbbold}{bbold}
\newcommand{\vect}[1]{\mathbbold{#1}}

\newtheorem{theorem}{Theorem}
\newtheorem{remark}[theorem]{Remark}
\newtheorem{example}[theorem]{Example}
\newtheorem{lemma}[theorem]{Lemma}
\newtheorem{definition}[theorem]{Definition}
\newtheorem{proposition}[theorem]{Proposition}
\newtheorem{problem}{Problem}

\newcommand{\x}{{\bf x}}
\newcommand{\ub}{{\bf u}}
\newcommand{\y}{{\bf y}}
\newcommand{\Lc}{{\mathcal L}}

\newcommand{\V}{{\mathcal V}}
\newcommand{\Ec}{{\mathcal E}}

\newcommand {\be}{\begin{equation}}
\newcommand {\ee}{\end{equation}}
\begin{document}

\begin{frontmatter}
 
\title{\LARGE\bf  Data-based control of Logical Networks}

%\title{
%%\LARGE
%\bf  On an extension of the Friedkin-Johnsen model: \\ 
%%The effects of a homophily-based influence matrix
%{\color{black} The role of homophily in the mutual influence process}}
%% Title, preferably not more than 10 words.

\thanks[footnoteinfo]{Corresponding Author.
%%E-mail addresses:  \texttt{giorgia.disaro@phd.unipd.it} (G. Disar\`o), \texttt{meme@dei.unipd.it} (M.E. Valcher).
}

\author[Padova]{Giorgia Disar\protect{\`o}}\ead{disarogior@dei.unipd.it},    % Add the 
\author[Padova]{Maria Elena Valcher\thanksref{footnoteinfo}}\ead{meme@dei.unipd.it}              % e-mail address 
%\author[Baiae]{Publius Maro Vergilius}\ead{vergilius@culture.ir}  % (ead) as shown

\address[Padova]{ Dipartimento di Ingegneria dell'Informazione
 Universit\`a di Padova, 
    via Gradenigo 6B, 35131 Padova, Italy}  % Please supply                                              
%\address[Rome]{Senate House, Rome}             % full addresses
%\address[Baiae]{The White House, Baiae}        % here.

%\author{Giorgia Disar\`o and Maria Elena Valcher
%\thanks{G. Disar\`o and  M.E. Valcher are with the Dipartimento di Ingegneria dell'Informazione,
% Universit\`a di Padova,
%    via Gradenigo 6B, 35131 Padova, Italy, e-mail:  \texttt{giorgia.disaro@phd.unipd.it, meme@dei.unipd.it}}
%   }
%  \date{}
%
%\begin{document}
%\maketitle


\begin{abstract}
  
 In recent years, data-driven approaches have become increasingly  pervasive across all areas of control engineering. However, the applications of data-based techniques to Boolean control networks (BCNs) are still very limited. 
 %Indeed, identifying a BCN from data is computationally heavy, {\color{black} since it requires identifying all   possible transitions that can occur in the network, and their number is typically} huge. 
 In this paper we 
  aim to fill this gap, by exploring the possibility of {\color{black} evaluating some basic features, i.e., reachability and equilibria, and  of solving two} fundamental control problems, i.e., safe control  and output regulation, for a BCN, leveraging only a limited amount of data  
 %  based only on data 
  generated by the network, without knowing or identifying its  model.
 \end{abstract}

\end{frontmatter}

\section{Introduction} \label{intro}

 Boolean networks (BNs) are discrete-time (autonomous) dynamical systems whose describing variables are constrained to take values in  %the finite alphabet consisting of 
the Boolean set. When their evolution is controlled by an external signal (an input), whose entries in turn take Boolean values,  they are called Boolean control networks (BCNs).
The interest in this class of logical systems originated in the 1960s, when Stuart Kauffman introduced them to model and analyze gene regulatory networks \cite{Kauffman}, and even if this 
may still be regarded as their most successful application area,   they have been adopted   in a wide range of fields, including systems biology, computer science, game theory, smart cities and network theory (see, e.g., \cite{born08,ChengGames,Green,SmartHome,FabioHealth}). 
 
The algebraic approach 
to BNs and BCNs,
proposed by Cheng et al. \cite{in_state_approach_bn,lin-rep-dyn-boolean,bn_statespace,BCNCheng},  and deeply relying on the concept of semi-tensor product (STP) \cite{STP2001}, 
has  made it possible to  represent BNs and BCNs by means of discrete-time state space models  in which the state, input and output are logical vectors, and the describing matrices are logical matrices.
%, in turn. 
%whose describing variables and matrices are constrained to be logical vectors and logical matrices, respectively. 
In this way, concepts and tools developed for linear time-invariant state space models have been  adjusted to formalize
and solve  control problems for BNs and BCNs   (see, e.g., \cite{ACC2014_BCN,LiYangChu2013,ZhangZhang}).%; however, this is far from an exhaustive review of the literature on BCNs). 

 Despite the simplifying  ``on/off" (``active/inactive",  ``high/low") logic at the basis of BCNs,
the size and complexity of
the  physical systems  they model (e.g., gene regulatory networks)
 make it difficult to derive an accurate
% However, BCNs  can be extremely complex   and exhibit a  large number of  possible transitions, and hence it may be difficult to have a precise 
model of the network, as well as to obtain it by means of identification techniques.
% Indeed, to uniquely identify a BCN, we need to know all the transitions that can occur in the network, and the inputs that induce them. Therefore, it is highly likely that, as the size of the network increases, the number of required data becomes prohibitive. 
 On the other hand, in recent times {\em direct} data-driven methods   have  gained interest in the control community, given their 
attractive feature of bypassing the identification phase,  to directly design the problem solution based on raw data. 
Yet, the use of partial data to solve control problems for logical networks is still at an early stage. 
%, and to the best of our knowledge 
 The first attempt in this direction has been made in \cite{BCNdd}, where a one-step data-driven approach to solve the output stabilization problem is proposed. More specifically, the authors in \cite{BCNdd} first reconstruct from the available data the output prediction matrix that maps the current input/output pair to the set of all possible next outputs, and then, based on this map, design an output feedback law to stabilize the output of the BCN to a desired value. However, in general, the set of all possible next outputs contains more than one element and hence the output prediction is possibly non-deterministic. {\color{black} To deal with this uncertainty, the authors draw inspiration from results on the stabilization of Probabilistic Boolean networks (PBNs).}
 % This uncertainty is handled by   resorting to Probabilistic Boolean networks (PBNs) to compensate for the fact 
 % that data are not  sufficient to identify the whole BCN and uniquely solve the stabilization problem.
 
{\color{black} The authors in \cite{Li_BCNDD}, instead,  exploit the informativity framework, introduced by van Waarde and coauthors in \cite{informativity} (see, also, \cite{info_magazine}), to investigate BCN control problems in a data-driven set-up.} The informativity framework is based on the idea that, when data are not informative enough to uniquely identify the system, the best we can do is impose that all the systems compatible with the collected data (among which there is surely also our target system) possess the desired property or behave in a certain way. {\color{black} In \cite{Li_BCNDD}  fundamental system properties such as controllability and stabilizability are investigated, and  the problems of state feedback stabilization and guaranteed cost control are addressed for BCNs. For both problems, necessary and sufficient conditions under which the data are informative for control design are provided.} 
%e will show that also}
%some standard control problems for BCNs can still be solved by using only a limited amount of data and without identifying the network. In \cite{BCNdd} a one-step data-driven approach to solve the output stabilization problem is proposed, in which, by avoiding the identification step, the required amount of data is reduced with respect to the traditional two steps procedures. However, the authors in \cite{BCNdd} resort to Probabilistic Boolean Networks (PBNs) to compensate for the lack of knowledge, due to the fact that the collected data are not sufficient to determine all  possible transitions in the network. To the best of our knowledge, no other paper in the literature has   addressed a control problem for logical networks based on partial data.

In this paper,  we also exploit this approach to address the limitations arising from having only a limited amount of data available. More specifically, we solve two fundamental control problems, namely safe control   and output regulation of a BCN via state feedback, by solving them for all logical systems compatible with the available data. This automatically guarantees that the problem is solved for the actual BCN, even under partial information.
%and at the same time allows us to remain within a deterministic setup. 
We provide necessary and sufficient conditions on the collected data to drive the evolution of a target BCN towards a safe set, and to set its output to a desired value, by means of a state feedback. Moreover,  for each of these problems, we   propose a   procedure to design such a feedback matrix based on the collected data. 
{\color{black} In order to  solve these problems, we   preliminarily investigate
under what conditions the reachability of a set of states  and the existence of equilibrium points can be inferred, for all the BCNs compatible with the data, directly from the  data themselves, since these notions prove to be instrumental in constructing the solutions. 
}

The paper is organized as follows. {\color{black} In Section \ref{prel} we briefly introduce Boolean networks and Boolean control networks. %contains some preliminary notions and results about BNs and BCNs. 
In Section \ref{datacoll_infoid} we explain the data collection process and we provide the notion of informativity for identifiability, along with its characterization. Section \ref{reach} contains an informativity-based characterization of the reachability property for a BCN. Section \ref{eqpts} reviews the notion of equilibrium points in BCNs and provides a systematic approach to identifying all equilibrium points consistent with the data.
In Section \ref{sec.SC} we formalize and solve the safe control problem for a BCN using only finite sequences of data. In Section \ref{p2}  the data-driven output regulation problem by state feedback for a BCN is investigated.} Section \ref{concl} concludes the paper. 
\smallskip

{\bf Notation.}  Given two nonnegative integers $k$ and $n$, with $k\le n$, we denote by $[k,n]$ the set of integers $\{k,k+1, \dots, n\}$. 
We consider Boolean vectors and matrices whose entries take values in $\mathcal B \triangleq \{0,1\}$, equipped with the usual logical operations, namely sum (OR)  $\vee$, product (AND) $\land$, and negation (NOT) $\neg$.  We denote by $\delta_k^i$ the $k$-dimensional $i$-th canonical vector, by $\mathcal L_k$ the set of all $k$-dimensional canonical vectors, and by $\mathcal L_{k\times q} \subset \mathcal B^{k\times q}$ the set of all $k\times q$ logical matrices, namely matrices whose $q$ columns are canonical vectors of size $k$. Any matrix $L\in \mathcal L_{k\times q}$ can be described as $L = [\delta_k^{i_1} \ \delta_k^{i_2} \ \dots \ \delta_k^{i_q}]$, for some indices $i_1,i_2,\dots, i_q\in [1,k]$.   $I_n$ denotes the $n$-dimensional identity matrix. 
The $k$-dimensional vector with {\color{black} all unitary (all zero) entries is denoted by $\vect{1}_k$ (by $\vect{0}_k$)}, and the $k\times l$ matrix with 
all zero entries by $\vect{0}_{k\times l}$. 
{\color{black} When clear from the context or unnecessary, the suffixes are omitted.} The $(i,j)$-th entry of a matrix $M$ is denoted by $[M]_{ij}$, and the $j$-th entry of a vector ${\bf v}$ by $[{\bf v}]_j$. The $j$-th column of matrix $M$ is denoted by ${\rm col}_j(M)$. We denote by $\odot$ the {\em Hadamard (or entry-wise) product}. 
{\color{black}  Given two Boolean matrices $M_1$ and $M_2$ of the same size, condition $M_1 \le M_2$ means that $[M_1]_{ij} \le [M_2]_{ij}$ for every $i,j$. Given a  vector  ${\bf v}\in {\mathbb R}^k$, we denote by ${\rm diag}({\bf v}^\top)$ the $k\times k$ diagonal matrix whose $(i,i)$-th entry is $[{\bf v}]_i$. }
%Given matrices $M_1, \dots, M_k$, the notation ${\rm blkdiag}(M_1, \dots, M_k)$ denotes the block diagonal matrix whose diagonal blocks are $M_1, \dots, M_k$.}
%A {\em permutation matrix} $P$  of size $k$ is a nonsingular square matrix in $\mathcal L_{k\times k}$. 
 %A $k \times k$ {\em cyclic (permutation) matrix} is a special permutation matrix that shifts the elements of a vector circularly. 
%a matrix described as follows 
%$$
%P = C = \begin{bmatrix} \delta_k^2 & \delta_k^3 & \dots & \delta_k^k & \delta_k^1\end{bmatrix}. 
%$$
A Boolean (in particular, a logical) square matrix $L$ of size $k$ 
is {\em reducible} if there exists a permutation matrix $P\in \mathcal L_{k\times k}$ such that
$$P^\top L P = \begin{bmatrix} L_{11} & L_{12}\cr \vect{0} & L_{22}\end{bmatrix},$$
where $L_{11}$ and $L_{22}$ are square Boolean (logical) matrices.
Otherwise, $L$ is  called {\em irreducible}.
$L$ is   irreducible if and only if
%$\bigvee_{i = 0}^{k-1}(L)^i = \
$I_k \vee L \vee \dots \vee L^{k-1}$ has all unitary entries.
 \\
There is a bijective correspondence between Boolean variables $X\in \mathcal B$ and vectors $\x \in \mathcal L_2$, defined by $\x = [X \ \neg X]^\top$. 
%{\color{red} verificato - Chapter 3 libro Chen} 
Given  
 $A \in \mathbb R^{m\times n}$ and $B\in \mathbb R^{p\times q}$,   their  {\em (left) semi-tensor product} ($\ltimes$) (see \cite{STP2001}) is
$$
A \ltimes B \triangleq (A \otimes I_{l/n})(B\otimes I_{l/p}), \quad l \triangleq {\rm l.c.m.}\{n,p\},
$$
{\color{black} where the symbol $\otimes$ denotes the {\em Kronecker product}.}
For the properties of the semi-tensor product, we refer to \cite{BCNCheng}.\\
{\color{black} Given two matrices $C\in {\mathcal B}^{m\times n}$ and $D\in {\mathcal B}^{p\times n}$ with the same number of columns, we denote by $*$ their {\em Khatri-Rao product}, which is defined as a column-wise semi-tensor product, or, equivalently, a column-wise Kronecker product, i.e.,
\begin{eqnarray*}
C*D &=& 
\left[\begin{array}{c|c|c} {\rm col}_1(C) \ltimes {\rm col}_1(D) & \dots & {\rm col}_n(C) \ltimes {\rm col}_n(D)\end{array}\right] \\
&=& \left[\begin{array}{c|c|c} {\rm col}_1(C) \otimes {\rm col}_1(D) & \dots & {\rm col}_n(C) \otimes {\rm col}_n(D)\end{array}\right]. 
\end{eqnarray*}
}
By exploiting the semi-tensor product, it is possible to extend the previously mentioned bijective correspondence to a bijective correspondence between $\mathcal B^n$ and $\mathcal L_{2^n}$, as follows.  Given a vector $X = [X_1 \ X_2 \ \dots \ X_n]^\top \in \mathcal B^n$, set 
$$\x \triangleq \begin{bmatrix} X_1 \\ \neg X_1\end{bmatrix} \ltimes \begin{bmatrix} X_2 \\ \neg X_2\end{bmatrix} \ltimes \dots \ltimes \begin{bmatrix} X_n \\ \neg X_n\end{bmatrix}.$$
It is possible to associate with a Boolean (in particular, a logical) matrix $L \in {\mathcal B}^{k\times k}$ a {\em directed graph} (or, {\em digraph}), $\mathcal D(L) = (\mathcal V, \mathcal E)$, where $\V =[1,k]$ is the set of nodes and $\Ec$ is
 the set of edges. An edge $(j,l)$, from $j$ to $l$, belongs to $\mathcal E$ if and only if $[L]_{lj} = 1$. 
A sequence $j_1 \to j_2 \to \dots \to j_r \to j_{r+1}$ in $\mathcal D(L)$ is a {\em path} of length $r$ from $j_1$ to $j_{r+1}$ if $(j_1,j_2), \dots, (j_r,j_{r+1})$ are   %distinct 
edges in $\mathcal E$. A closed path is a {\em cycle}. 
%If a cycle $\gamma$ has no repeated edges is called {\em elementary}, and its length coincides with the number of its distinct nodes, and is denoted by $|\gamma|$. 
A node $j^*\in {\mathcal V}$ is said to be {\em globally reachable} if there exists a path to $j^*$ from any other node in the network. 
A digraph $\mathcal D(L)$  is {\em strongly connected} if all its nodes are globally reachable, and this is the case if and  only if
 $L$ is irreducible. 
\bigskip

{\color{black} \section{Introduction to Boolean control networks} \label{prel}}
 
 A {\em Boolean control network (BCN)} is a logical system described by the following equations
\begin{subequations} \label{bacon}
\begin{eqnarray}
X(t+1) &=& f(X(t), U(t)),   \label{bcn}\\ 
Y(t) &=& h(X(t)), \qquad \qquad t \in \mathbb Z_+, \label{bcno}
\end{eqnarray}
\end{subequations}
where $X(t)\in \mathcal B^n$ is the $n$-dimensional state variable, $U(t)\in \mathcal B^m$ is the $m$-dimensional input, and 
$Y(t)\in \mathcal B^p$ is the $p$-dimensional output  at time $t$. $f$ and $h$  are logical functions, i.e., $f: \mathcal B^n \times\mathcal B^m \to \mathcal B^n$, while $h: \mathcal B^n \to \mathcal B^p$. 
If the logical system is autonomous, namely no input acts on it, then the BCN becomes a {\em Boolean network} (BN) and equation \eqref{bcn} becomes
\be \label{bn} 
X(t+1) = f(X(t)), \quad t \in \mathbb Z_+, 
\ee
where %$X(t)\in \mathcal B^n$ and   
$f: \mathcal B^n \to \mathcal B^n$ is a  logical function.
By exploiting the bijective correspondence between Boolean  and logical vectors, 
the BN \eqref{bn} 
 can be equivalently expressed via its {\em algebraic representation} \cite{BCNCheng} as 
 \be
 \label{bna} 
 \x (t+1) = L \x (t), \quad t\in \mathbb Z_+,
\ee
where $\x(t) \in \mathcal L_{N}$ and $L\in {\mathcal L}_{N\times N}$, with $N \triangleq 2^n$.
Similarly, the algebraic representation  of  
the BCN  \eqref{bacon} is
 \begin{subequations}\label{BCNtot}
\begin{eqnarray} 
 \label{bcnA} 
\x (t+1) &=& L \ltimes \ub(t) \ltimes \x(t), %= L\ub(t)\x(t), 
 \quad t\in \mathbb Z_+,\\
 \y(t) &=& H\x(t), \label{bcnAo}
 \end{eqnarray}
 \end{subequations}
where $\x(t) \in \mathcal L_{N}$, $\ub(t) \in \mathcal L_{M}$, $\y(t) \in \mathcal L_{P}$, $L\in \mathcal L_{N\times NM}$ and $H\in {\mathcal L}_{P\times N}$, with $N \triangleq 2^n$, $M \triangleq 2^m$ and $P\triangleq 2^p$.
It is worth remarking that {\color{black}   model \eqref{BCNtot} can be used to represent any state-space model whose state, input and output variables take values in sets of finite cardinalities $N$, $M$ and $P$, respectively, and hence} all the subsequent analysis holds for generic nonnegative integer values of $N$, $M$ and $P$, not necessarily powers of $2$. 
The matrix $L\in \mathcal L_{N\times NM}$, whose columns are canonical vectors of size $N$, can be divided into $M$ square blocks of size $N$ as follows 
\be \label{Lblock}
L = \begin{bmatrix}
L_1 &\vline& L_2 &\vline& \dots &\vline& L_M
\end{bmatrix},
\ee
where each block $L_i\in \mathcal L_{N\times N}, i \in [1,M]$, describes the behavior of   the $i$-th subsystem of the BCN, i.e., the Boolean network obtained when $\ub(t) = \delta_M^i, \forall t \in \mathbb Z_+$:
\be \label{Li}
\x (t+1) = L_i \x (t), \quad t\in \mathbb Z_+.
\ee 

%Indeed, a BCN can be seen as a Boolean switched system consisting of $M$ subsystems, one for each possible value of the input $\ub$. At every time instant we 
%We can associate with each matrix $L_i, i \in [1,M]$, a digraph that we denote by $\mathcal D(L_i) \triangleq (\mathcal V, \mathcal E_i)$, where $\mathcal V = [1,N]$ is the set of nodes and $\mathcal E_i$ is the set of edges. An edge $(j,l)$ belongs to $\mathcal E_i$ if and only if $[L_i]_{lj} = 1$.  

% \medskip
{\color{black} \section{Data collection process and informativity for identifiability of BCNs} \label{datacoll_infoid}} %Data-driven stabilization %of a BCN 
%by state feedback

%In this section we restrict our attention to the state equation \eqref{bcnA} of the BCN.
We assume to have performed {\color{black} some  offline experiments\footnote{{\color{black} It is worth noticing that the number and the lengths of the experiments will play no role in the following analysis.} {\color{black} However, a higher number of experiments  may allow  to obtain information that could not be collected in a single one.}}, say $r\ge 1$, during which we have collected   state/input/output data from the BCN \eqref{BCNtot} on finite time intervals $[0,T_i], T_i \in \mathbb Z_+, i \in [1,r]$. }
%\footnote{
% Instead of %It is not mandatory to perform 
% a {\em single} experiment to collect data, we can   perform {\em several} experiments, say $r$, on some finite time intervals $[0,T_i], i \in [1,r]$. It suffices to arrange data consistently  with the way it is illustrated in the case of a single experiment, {\color{black} by this meaning that for every $j\in [1,T]$, $T\triangleq \sum_{i=1}^r T_i,$ the $j$-th column of $X_f$ represents the successor of the vector in the $j$-th column of $X_p$ when the input is $U_p\delta^j_T$. It is worth remarking that }},
%with $T$ sufficiently large (in a sense that will be clarified in the sequel), and 
{\color{black} We define the   vector sequences $\x_d^i \triangleq \{\x_d^i(t)\}_{t = 0}^{T_i}$, $\ub_d^i \triangleq \{\ub_d^i(t)\}_{t = 0}^{T_i-1},$ and $\y_d^i \triangleq \{\y_d^i(t)\}_{t = 0}^{T_i-1}, i \in [1,r]$, and accordingly $\x_d \triangleq \{\x_d^i\}_{i=1}^r$, $\ub_d \triangleq \{\ub_d^i\}_{i=1}^r$ and $\y_d \triangleq \{\y_d^i\}_{i = 1}^{r}$.} 
%To ease the notation, in the following we will simply indicate these sequences as $\x_d$ and $\ub_d$. 
{\color{black} We rearrange the  data collected during the $r$ experiments into the following logical matrices with $T \triangleq \sum_{i = 1}^{r}{T_i}$ columns:
\begin{eqnarray*} \label{data} 
X_p &\triangleq& \left[\begin{array}{ccc|ccc}
\x_d^1(0) & \dots & \x_d^1(T_1-1) & \x_d^2(0) & \dots & \x_d^r(T_r-1)
\end{array}\right], % \in \mathcal L_{N\times T},   
\label{Xp}\\
X_f &\triangleq& \left[\begin{array}{ccc|ccc}
\x_d^1(1) & \dots & \x_d^1(T_1) & \x_d^2(1) & \dots & \x_d^r(T_r)
\end{array}\right],    
\label{Xf}\\
U_p &\triangleq& \left[\begin{array}{ccc|ccc}
\ub_d^1(0) & \dots & \ub_d^1(T_1-1) & \ub_d^2(0) & \dots & \ub_d^r(T_r-1)
\end{array}\right],    
\label{Up} \\
Y_p &\triangleq& \left[\begin{array}{ccc|ccc}
\y_d^1(0) & \dots & \y_d^1(T_1-1) & \y_d^2(0) & \dots & \y_d^r(T_r-1)
\end{array}\right],    
\label{Yp}
\end{eqnarray*}
where the subscripts $p$ and $f$ stand for past and future, respectively. 
%We also introduce the following global data matrices 
%\begin{eqnarray} \label{data}
%X_p &\triangleq& \begin{bmatrix}
%X_p^1 & X_p^2 & \dots & X_p^r
%\end{bmatrix} \in \mathcal L_{N\times T}, \ \  \label{Xp}\\
%X_f&\triangleq&\begin{bmatrix}
%X_f^1 & X_f^2 & \dots & X_f^r
%\end{bmatrix} \in \mathcal L_{N\times T}, \ \ \label{Xf}\\
%U_p &\triangleq& \begin{bmatrix}
%U_p^1 & U_p^2 & \dots & U_p^r
%\end{bmatrix} \in \mathcal L_{M\times T}, \  \ \label{Up}
%\end{eqnarray}
%where 
 }
%, and we analyze under what conditions it is still possible to solve control problems using only these data.\\

 {\color{black} As every BCN \eqref{BCNtot} bijectively corresponds to its describing matrices, we define the {\em set of Boolean control networks \eqref{BCNtot} compatible with the data} $(\x_d$, $\ub_d, \y_d)$ as 
 \begin{align} \label{compBCNo}
\!\!\!\!\mathcal B_d \!& \triangleq  \big\{(\tilde L,\tilde H)\in \mathcal L_{N\times NM} \times \mathcal L_{P\times N} : \nonumber\\
\!\! &\!\!\! \x_d^i(t+1) = \tilde L \ltimes \ub_d^i(t)\ltimes \x_d^i(t), \
 \y_d^i(t) = \tilde H \x_d^i(t), \nonumber\\
 \!\! &\!\!\!\forall t \in [0,T_i-1], \forall i \in [1,r] 
\big\}.  \!\!\!
\end{align}
Obviously, the pair $(L,H)$   corresponding to the BCN  \eqref{BCNtot} 
that generated the data belongs to $\mathcal B_d$.
However, ${\mathcal B}_d$ may also include additional matrix pairs when the available data do not uniquely determine $(L,H)$. 
% is simply the singleton  $\{L\}$, or contains other matrices, we cannot or not uniquely identify the BCN. 
To this end, we introduce the notion of {\em informativity for identifiability}\footnote{{\color{black} 
As a matter of fact,
what we call informativity for {\em identifiability} is referred to as informativity for {\em model reconstructability} in
Definition 3.3 of \cite{Li_BCNDD}.  
%introduces the concept of informativity for {\em model reconstructability}, the authors use the term {\em identifiability} with a slightly different meaning, while they use the term {\em model reconstructability} to refer to what this paper defines as identifiability.
}}, that extends Definition 3.3 in \cite{Li_BCNDD}   to  BCNs  as in \eqref{BCNtot}, which also include an output equation. %Consequently, the output data must be incorporated into the formulation.
}
 \smallskip
 
\begin{definition} \label{info_id_def_o}
 The  data $(\x_d, \ub_d, \y_d)$ are called {\em informative for identifiability} of the BCN \eqref{BCNtot} if ${\mathcal B}_d \equiv \{(L,H)\}$. 
\end{definition}
\smallskip

{\color{black} Despite the presence of the output equation, the characterization of informativity for identifiability is the same as the one given in Proposition 3.4 of \cite{Li_BCNDD} for BCNs described only by \eqref{bcnA}. Lemma \ref{info_id} proves such equivalence and provides  an additional  equivalent condition.} 
\smallskip

 \begin{lemma} \label{info_id}
{\color{black} The following facts are equivalent. 
\begin{itemize}
    \item[i)]
The data $(\x_d$, $\ub_d, \y_d)$ are informative for identifiability of the BCN \eqref{BCNtot}.
\item[ii)] For every $(i,j) \in [1,M]\times [1,N]$, there exists $k\in[1,T]$ such that
\be
\begin{bmatrix} 
X_p \cr 
U_p 
\end{bmatrix} \delta_T^k = 
\begin{bmatrix}
\delta_N^j \cr 
\delta_M^i
\end{bmatrix}.
\label{ijk}
\ee
\item[iii)] The   
logical matrix {\color{black} $U_p*X_p$
has no zero rows (equivalently,} is of full row rank). 
\end{itemize}
}
\end{lemma}

{\color{black}
\begin{proof}
%The equivalence of $i)$ and $iii)$ has been proved in \cite{Li_BCNDD} for a BCN described only by equation \eqref{bcnA}. Nonetheless, the proof can be readily adapted to a BCN that include the output equation, as in \eqref{BCNtot}.
{\em i)} $\Rightarrow$ {\em iii)}\ If {\em i)} holds, then the data $(\x_d$, $\ub_d)$ are informative for identifiability of   equation \eqref{bcnA}, and hence, by Proposition 3.4 in \cite{Li_BCNDD}, we can claim that 
{\em iii)} holds. \\
{\em iii)} $\Rightarrow$ {\em i)}\ If condition {\em iii)} holds, by Proposition 3.4 in \cite{Li_BCNDD}, we can uniquely identify the matrix $L$. 
On the other hand,  if the matrix $U_p*X_p$
 is of full row rank, then 
 %we deduce from the collected data every possible state transition that can occur in the network. But this  implies that 
 $X_p$ is of full row rank, in turn.
This means that all   possible states appear at least once in $X_p$, and consequently also the corresponding output in $Y_p$.   
 This allows to uniquely identify  also the matrix $H$ in \eqref{bcnAo}, and hence the data $(\x_d, \ub_d, \y_d)$ are informative for identifiability of the whole BCN \eqref{BCNtot}.

{\em ii)} $\Leftrightarrow$ {\em iii)} is straightforward.
\end{proof}
}

%  As a BCN \eqref{bcnA}   bijectively corresponds to its describing   matrix, we define the {\em set of Boolean Control Networks \eqref{bcnA} compatible with the data} $(\x_d$, $\ub_d)$ as 
% \begin{eqnarray} \label{compBCN}
% {\mathcal B}_d &\triangleq& \{\tilde L \in \mathcal L_{N\times NM} :  \\
%  &&\x_d(t+1) = \tilde L \ltimes \ub_d(t) \ltimes \x_d(t),  \forall t \in [0,T-1]\}. \nonumber
% \end{eqnarray}
% %With some abuse of terminology, in the following we will say that a BCN is in ${\mathcal B}_d$ if its describing matrix  $\tilde L$ is in ${\mathcal B}_d$. 
% Obviously, the matrix $L$ of the BCN that has generated the data belongs to ${\mathcal B}_d$. However, %the set 
% ${\mathcal B}_d$ may also contain other matrices if the data do not allow to uniquely determine  $L$. 
% % is simply the singleton  $\{L\}$, or contains other matrices, we cannot or not uniquely identify the BCN. 
% To this end, we introduce the notion of {\em informativity for identifiability}. 
% \smallskip

% \begin{definition} \label{info_id_def}
% We say that the data $(\x_d$, $\ub_d)$ are {\em informative for identifiability} of the BCN \eqref{bcnA} if ${\mathcal B}_d \equiv \{L\}$. 
% \end{definition}
%  \smallskip

% %, whose proof is omitted since it is straightforward. 
  \smallskip

 \begin{remark}
%Note that i
{\color{black}  If the data $(\x_d, \ub_d, \y_d)$ are informative for identifiability and 
condition \eqref{ijk} holds, then one can identify the $j$-th column of matrix $L_i$ as
${\rm col}_j (L_i) = X_f\delta_T^k$.
As it can be deduced from Remark 3.5 in \cite{Li_BCNDD},
one can identify the whole matrix $L$ using
the following formula
$$
L = X_f \odot_{\mathcal B} (U_p*X_p)^\top,
$$
where $\odot_{\mathcal B}$ is the Boolean product of matrices (that acts as the regular matrix product with the multiplication replaced by the AND and the addition replaced by the OR).
 Alternatively, by resorting to the standard matrix product, one can deduce $L$ as follows:
$$
L = X_f(U_p*X_p)^\top \left({\rm diag}(\vect{1}_N^\top X_f(U_p*X_p)^\top)\right)^{-1}.
$$
Finally, as $X_p$ is of full row rank (see the proof of Lemma \ref{info_id}), it is immediate to obtain  $H$ as 
$$
H = Y_pX_p^{\#},
$$
where $X_p^{\#} \triangleq X_p^\top(X_pX_p^\top)^{-1}$ is a %(specific) 
right inverse of $X_p$. 
}
\end{remark}
\smallskip

To make the following analysis meaningful, from now on we assume that the collected data are {\em not} informative for identifiability of the considered BCN, namely
${\mathcal B}_d \supsetneq \{(L,H)\}$,
% they do not allow to uniquely identify the matrix $L$, 
and we investigate under what conditions it is still possible to use them to {\color{black} evaluate certain properties or to} solve some control problems. 
{\color{black} By adhering to the approach initiated in the papers by van Waarde, Eising et al. (see, e.g., \cite{informativity,info_magazine}) and first explored for logical networks in \cite[Definition 2.1]{Li_BCNDD}, when the available data do not allow identifying the BCN that generated them, the best we can do is to understand if a certain property holds or some control problem is solvable for all the BCNs in ${\mathcal B}_d$. When so, we will say that the data $(\x_d, \ub_d, \y_d)$ are informative for the property or for the problem solution.}

{\color{black} \section{Reachability property} \label{reach}}

{\color{black} In this section, we provide an informativity-based characterization of the reachability of a set of states for a BCN. Before analyzing the problem in a data-driven context, we first recall the model-based approach to the reachability property.
\smallskip
}

\begin{definition} \cite{BCNCheng} \label{reach}
Given a BCN as in \eqref{bcnA}, we say that a state $\x_f = \delta_N^i$ is {\em reachable} from $\x_0 = \delta_N^j$ if there exist $\tau \in \mathbb Z_+$ and an input $\ub (t), t \in [0, \tau-1],$ that leads the state trajectory from $\x(0) = \x_0$ to $\x(\tau) = \x_f$. 
{\color{black} A set ${\mathcal X}\subset {\mathcal L}_N$ is {\em reachable} from $\x_0$, if there exists $\x_f \in \mathcal X$ that is reachable from $\x_0$. 
%from $\x_0$ one can reach at least one of the states in ${\mathcal X}$. 
A state $\x_f$ (a set ${\mathcal X}$) is {\em globally reachable} if it is reachable from every $\x_0$.
The     BCN  \eqref{bcnA} is {\em reachable} if    every state $\x_f$ is globally reachable. }
\end{definition}
 \smallskip
 
 {\color{black} The  reachability of a single state, of a set ${\mathcal X}$ and of the whole BCN can be characterized (see, e.g., Section 3 in \cite{EF_MEV_BCN_Aut} and Chapter 16 in \cite{BCNCheng})}  by resorting to the matrix  
\be \label{ltot}
L_{tot} \triangleq  L_1 \vee L_2 \vee \dots \vee L_M \in {\mathcal B}^{N\times N}, 
\ee
and the associated digraph $\mathcal D(L_{tot}) \triangleq (\mathcal V, \mathcal E_{tot})$.
%where $\mathcal E_{tot} = \bigcup_{i\in [1,M]} \mathcal E_i$,
 The state $\x_f = \delta_N^i$ is reachable from $\x_0 = \delta_N^j$ if and only if there exists $\tau \in \mathbb Z_+$ such that $[L_{tot}^\tau]_{ij} = 1$, or equivalently,  there is a path of length $\tau$ from $j$ to $i$ in $\mathcal D(L_{tot})$.\footnote{In the sequel, we will interchangeably write that a node (in the digraph) or a state (of the BCN) is  reachable from another node or state. {\color{black} More generally, as there is a bijective correspondence between states and nodes, with a slight abuse of terminology, we will often refer to state $\delta^i_N$ to denote node $i$ in the digraph, and vice versa. }} 
 {\color{black} This is the case if and only if
 $(\delta^i_N)^\top \left(\sum_{\tau=0}^{N-1} L_{tot}^\tau\right) \delta^j_N > 0$.
 Consequently, the state $\x_f = \delta^i_N$ (the set 
 ${\mathcal X}= \{ \delta^{i_1}_N, \delta^{i_2}_N, \dots, \delta^{i_d}_N\}$) is globally reachable if and only if  all entries of the row vector $(\delta^i_N)^\top \left(\sum_{\tau=0}^{N-1} L_{tot}^\tau\right)$
 (of the row vector $(\sum_{\ell=1}^d \delta^{i_\ell}_N)^\top \left(\sum_{\tau=0}^{N-1} L_{tot}^\tau\right)$) are positive.
 This implies that 
 a BCN is reachable if and only if the matrix $\sum_{\tau=0}^{N-1} L_{tot}^\tau$ has all positive entries, which means that} $L_{tot}$ is irreducible, or equivalently $\mathcal D(L_{tot})$  is strongly connected. \\

{\color{black}
We are now ready to introduce the notion of informativity for reachability of a set of states $\mathcal X$ that extends Definition 3.7 in \cite{Li_BCNDD}.
\medskip

 \begin{definition} \label{info_reach_def}
Given a set ${\mathcal X}\subset  \mathcal L_N$ (in particular, a state $\x_f\in {\mathcal L}_N$), we say that  the data $(\x_d, \ub_d, \y_d)$ are {\em informative for reachability of} $\mathcal X$ if the set $\mathcal X$ (the state $\x_f$) is globally reachable for all the BCNs in ${\mathcal B}_d$. 
\end{definition}
 \medskip

Since the reachability of some state $\x_f$ from another state $\x_0$ is related to the existence of a finite number of state transitions that lead the state from $\x_0$ to $\x_f$, 
%(equivalently, a path of finite length in $\mathcal D(L_{tot})$ from the node associated with $\x_0$ to the node $\x_f$), 
it is clear that when we try to verify this property for
all BCNs in ${\mathcal B}_d$, we can
rely only on  the state transitions that are revealed by the data. This  can be performed by first identifying the best possible approximation of $L_{tot}$ based on data, say $L_{tot}^d \in {\mathcal B}^{N\times N}$, and then verifying the reachability properties as explained in the model-based approach by referring to such matrix. Note that in general (differently from $L_{tot}$) the matrix $L_{tot}^d$ may have zero columns. This happens 
if there exists $j\in [1,N]$ such that $\delta^j_N$ does not appear as a column of $X_p$. 
We have the following result.
\medskip

\begin{proposition} \label{Ltot_data}
    Set $L_{tot}^d \triangleq X_f \odot_{\mathcal B} X_p^\top$. Then
    for every $(i,j) \in [1,N]\times [1,N]$, there exists $k\in[1,T]$ such that
\be
\begin{bmatrix} 
X_p \cr 
X_f
\end{bmatrix} \delta_T^k = 
\begin{bmatrix}
\delta_N^j \cr 
\delta_N^i
\end{bmatrix},
\label{statostato}
\ee
if and only if $[L_{tot}^d]_{ij}=1$. Therefore
the data are informative for reachability of the set 
 ${\mathcal X}= \{ \delta^{i_1}_N, \delta^{i_2}_N, \dots, \delta^{i_d}_N\}$
 (of the state $\x_f = \delta^i_N$)  if and only if  all entries 
 of the row vector $(\sum_{\ell=1}^d \delta^{i_\ell}_N)^\top \left(\sum_{\tau=0}^{N-1} (L_{tot}^d)^\tau\right)$ (of the row vector $(\delta^i_N)^\top \left(\sum_{\tau=0}^{N-1} (L_{tot}^d)^\tau\right)$) are positive.
\end{proposition}

\begin{proof}
    The first statement follows from Lemma 3.8 in \cite{Li_BCNDD}. The second part is immediate from the first statement and the comments before Proposition \ref{Ltot_data}.
\end{proof}
}
\smallskip

{\color{black} 
Proposition \ref{Ltot_data} provides an extension of  Theorem 3.9 in \cite{Li_BCNDD}, where a characterization of informativity for (the overall) reachability  is provided. Indeed, differently from \cite{Li_BCNDD}, we   focus on the reachability of a set or a state, since this will be useful in the following sections.
Moreover, as an interesting byproduct, we have deduced the matrix
$L_{tot}^d$ 
that represents the best possible
 ``under-approximation" of the matrix $L_{tot}$, in the sense that it is the   largest Boolean matrix  that satisfies $L_{tot}^d \le L_{tot}$ for all the matrices $L_{tot}$ of the BCNs in ${\mathcal B}_d$.  

It is worth underlining that Proposition \ref{Ltot_data} only provides an answer to the question of whether or not data are informative for reachability, but it does not offer any insight into the way one can select an input sequence that achieves the goal. To this end, we  introduce an algorithm, which also allows us to explore another concept that will be useful in the subsequent analysis. 
\medskip

\begin{definition} \label{basin}
Given a BCN described as in \eqref{bcnA}, the {\em basin of attraction} ${\mathcal S}(\x_f)$ {\em of a state} $\x_f\in {\mathcal L}_N$ is the set of states  $\x_0 \in {\mathcal L}_N$ from which it is possible to reach $\x_f$. 
\end{definition}
\medskip

Clearly, a state $\x_f$ is
  globally reachable if and only if its basin of attraction   ${\mathcal S}(\x_f)$ coincides with ${\mathcal L}_N$. 

  Algorithm \ref{alg1}, below, provides a way to determine
  the basin of attraction ${\mathcal S}^*$ of a set ${\mathcal X}$, by this meaning the union of the basins of attractions of the states in ${\mathcal X}$. For each state in ${\mathcal S}^*\setminus {\mathcal X}$ the algorithm also returns  
 a possible choice for the input that if applied to that state at a certain time instant $\bar t$, it ensures that at the next time instant $\bar t +1$ the distance between the state and the target set $\mathcal X$ decreases. Therefore, it readily follows that, by applying the inputs provided by the algorithm, if the set 
$\mathcal X$ is globally reachable, it will be reached in a finite number of steps.   
 %sequence that drives the state to ${\mathcal X}$ is given.
}
\medskip

\begin{algorithm}[h]
\caption{Reachability of ${\mathcal X}$} \label{alg1} 
\smallskip
\textbf{Input:}  - The data matrices $X_p, X_f$ and $U_p$;\\
\hspace*{3.18em} - The set ${\mathcal X}= \{ \delta^{i_1}_N, \delta^{i_2}_N, \dots, \delta^{i_d}_N\}$. \\
\textbf{Output:} - A possible choice of inputs $(\ub_1,\ub_2, \dots, \ub_N)$;
%Is  ${\mathcal X}$ globally reachable for every BCN in ${\mathcal B}_d$?\\
%\hspace*{4.5em} Yes/No. \\
\hspace*{4em} - The basin of attraction ${\mathcal S}^* \triangleq \bigcup_{\x \in {\mathcal X}} \mathcal S(\x)$ of \\
 \hspace*{4.9em}${\mathcal X}$.  \smallskip \\
{\em Initialization:} $(\ub_1,\ub_2, \dots, \ub_N)= (\vect{0}, \vect{0},  \dots, \vect{0})$. \\
%{\tt for} $r \in [1,N]$, {\tt do}\\
%\hspace*{1.5em}{\tt if} $\delta^r_N \in \mathcal X$, {\tt then}\\
%{\color{black} \hspace*{2.5em} $\ub_r$ $\leftarrow$ {\tt None}} \\
%\ub_r^*$ (see \eqref{input_safetosafe}) \\
%\hspace*{1.3em}  {\tt end if}\\
%{\tt end for}\\  
Set $d = 0$,  ${\mathcal S}_0 =  \mathcal X$.\\
%${\mathcal S}^* = \mathcal C^d(\y^*)$.\\
{\em Iterative procedure:}\\  {\tt if} ${\mathcal S}_d \ne \emptyset$, {\tt then}
\begin{itemize}
\item  $d \leftarrow d+1$
\item ${\mathcal S}_d = \emptyset$
\item {\tt for}   $k \in [1,T]$, {\tt do}\\
\hspace*{2.5em}{\tt if}
\hspace*{4em}$$
\begin{bmatrix} 
X_p \cr 
X_f 
\end{bmatrix} \delta_T^k = 
\begin{bmatrix}
\delta_N^r \cr 
\delta_N^q
\end{bmatrix}, \qquad \exists \delta^q_N\in {\mathcal S}_{d-1},
$$
\hspace*{2.5em}{\tt and} $\delta^r_N \not\in \cup_{i=1}^{d} {\mathcal S}_i$, {\tt then} \smallskip\\
\hspace*{4.5em}
${\mathcal S}_d \leftarrow {\mathcal S}_d \cup \{\delta^r_N\}$ and $\ub_r \triangleq U_p \delta_T^k$ \\
\hspace*{2.5em}{\tt end if} \\
{\tt end for}
\item Go back to the {\em Iterative procedure}.
\end{itemize}
{\tt else}  set $\mathcal S^* = \bigcup_{i=0}^{d-1} \mathcal S_i$ %and go to the {\em Conclusion}
\\
{\tt end if} \\
{\em Conclusion:}\ Return ($(\ub_1,\ub_2, \dots, \ub_N)$, $\mathcal S^*$)

\end{algorithm}

Note that Algorithm \ref{alg1} identifies first the set {\color{black} ${\mathcal S}_1$ of  states  in  ${\mathcal L}_N \setminus \mathcal X$ that have a successor in $\mathcal S_0 = \mathcal X$.  For each such state,  say $\delta^r_N$, it memorizes  in $\ub_r$ one of the values of the input that 
 allows the transition from $\delta^r_N$ to one of the states in $\mathcal X$, say $\delta^q_N$ (referring to the algorithm notation).}
  Then it identifies the set ${\mathcal S}_2$ of the  states that do not belong to ${\mathcal S}_0\cup {\mathcal S}_1 = \color{black}{\mathcal{X}} \cup {\mathcal S}_1$, but have a successor in ${\mathcal S}_1$.   Again, for each $\delta^r_N\in {\mathcal S}_2$, it memorizes one of the values of the input that 
 allows the transition from $\delta^r_N$ to some state   in ${\mathcal S}_1$.
 Since the set of states of a BCN is finite, there exists $d \ge 1$ such that ${\mathcal S}_d =\emptyset$.
{\color{black}  If at that stage $\bigcup_{i=0}^{d-1} {\mathcal S}_i = \mathcal L_N$, 
 this means that
 $\mathcal X$ is globally reachable.} %and
 %. As a byproduct, the algorithm returns also the set of states from which it is possible to reach $\bar \x$, i.e., 
%$\mathcal S^*$ is  its basin of attraction. }
\smallskip

% We are now ready to  state Lemma \ref{info_reach}, whose proof is straightforward from the previous analysis.
%   \smallskip
  
%  \begin{lemma} \label{info_reach}
% Given {\color{black} $\mathcal X \subset \mathcal L_N$}, the data $(\x_d, \ub_d, \y_d)$ are informative for reachability of {\color{black} $\mathcal X$} if and only if 
% %\begin{itemize}
% %\item[i)] The matrix $X_p$ is of full row rank\footnote{This implies that $T$ has to be at least equal to the state dimension $N$.}.
% %\item[ii)] 
% Algorithm \ref{alg1} terminates with an affermative answer.
% %\end{itemize}
% \end{lemma}
% \smallskip
 
{\color{black} \section{Equilibrium points of a BCN} \label{eqpts}}

{\color{black} We now recall the notions of equilibrium point and limit cycle in the context of logical systems. } 
\smallskip
 
\begin{definition} \label{eqpt}

Given a BN as in \eqref{bna},
%we say that 
a state $\x_e\in {\mathcal L}_N$ is an {\em equilibrium point of the BN} if $\x_e = L \x_e$. 
Given a BCN as in \eqref{bcnA}, we say that $\x_e$ is an {\em equilibrium point of the BCN 
corresponding to the input} $\ub=\delta^i_M$
if  $\x_e$ is an equilibrium point of its $i$-th subsystem  \eqref{Li}, 
 i.e.,  $\x_e = L \ltimes \delta^i_M \ltimes \x_e$.
\end{definition}

Note that  $\x_e=\delta^j_N$ is  an 
 equilibrium point of a BN or a BCN 
 if and only if the
    node $j$ has a    self-loop in ${\mathcal D}(L)$ or ${\mathcal D}(L_{tot})$, respectively.
    {\color{black} This amounts to saying that $[L]_{jj}$ or $[L_{tot}]_{jj}$, respectively, is unitary.}
 \smallskip

 {\color{black} Based on the comments about $L_{tot}$ in the previous section, it is immediate to deduce that the set of equilibrium points compatible with the available data, say, $\mathcal X_e^d$, namely the set of states that are equilibria in all the BCNs in $\mathcal B_d$, 
 coincides with the set 
 of states $\delta^j_N$ such that $[L_{tot}^d]_{jj} =1$, or, equivalently, $[X_f(X_p)^\top]_{jj} >0$, which is consistent with the 
analysis proposed in Section 3.2 of \cite{Li_BCNDD}. 
%More specifically, 
 % $\delta^j_N$ is an equilibrium point for all BCNs in ${\mathcal B}_d$ (corresponding to some input $\ub$) if and only if $[L_{tot}^d]_{jj} =1$.
 As an alternative, we can notice that $[L_{tot}^d]_{jj} =1$ if and only if  $\exists k \in [1,T]$ such that 
$$ 
\begin{bmatrix} 
X_p \cr 
X_f 
\end{bmatrix} \delta_T^k = 
\begin{bmatrix}
\delta_N^j \cr 
\delta_N^j
\end{bmatrix} \Leftrightarrow 
(X_p \odot X_f) \delta^k_T =  \delta_N^j,
$$ 
where we recall that $X_p \odot X_f$ is the Hadamard product of $X_p$ and $X_f$. This characterization has the advantage of allowing us to identify the input to which the equilibrium point corresponds as
}
$\ub_j \triangleq U_p \delta_T^k = \delta_M^i, \exists i \in [1,M]$. 
\smallskip 

 {\color{black}
 Since the number of possible states is finite,
 as time evolves, the state of a Boolean network either becomes constant (and hence reaches an equilibrium point) 
 or it becomes periodic. In the latter case, we talk about a limit cycle. Note that an equilibrium point is actually a limit cycle of unit length.
 } 
 %{\color{red} c'e' un'asimmetria perche' prima - per i punti di equilibrio - parliamo sia di BN che di BCN, mentre per i cicli limite solo di BNs. Forse varrebbe la pena focalizzarci solo sulle BNs, ma comunque dobbiamo recuperare in qualche modo sia per punti di equilibrio che per cicli limite l'informazione su quale ingresso permette la transizione desiderata.}
 \medskip
 
 \begin{definition}\cite{EF_MEV_BCN_Aut} \label{limitcycle}
Given a BN as in \eqref{bna}, an ordered sequence of distinct logical vectors $(\delta_N^{i_1}, \delta_N^{i_2}, \dots, \delta_N^{i_k})$ is   a {\em limit cycle $\mathcal C$} if, taken $\x(0) = \delta_N^{i_\ell}$ for some $\ell\in[1,k]$, the corresponding state evolution $\x(t)$ is periodic of period $k$ and, for every $t \in \mathbb Z_+$, $\x(t) = \delta_N^{i_j}$, with $j\in [1,k]$ satisfying $j = (t+\ell){\rm mod} k$. 
%A limit cycle $\mathcal C$ is said to be {\em globally attractive} if for every $\x(0) \in \Lc_N$ there exists $\tau \in \mathbb Z_+$ such that $\x(t)$ is a state belonging to $\mathcal C$ for every $t\in\mathbb Z_+, t\ge \tau$. 
\end{definition}
 \smallskip

{\color{black} We will make use of this concept in Section \ref{p2}.}

\section{Data-driven safe control} \label{sec.SC}

In \cite{EFMEV_JCD} the safe control of a BCN was first addressed,  then it was extended to handle random impulsive logical control networks in \cite{Zhou23} and Probabilistic Boolean control networks using an event-triggered approach in \cite{Shao25}.
The safe control  problem  is stated as follows. Given a BCN (\ref{bcnA}), suppose that the set of states ${\mathcal L}_N$ can be partitioned into a set of {\em unsafe states} ${\mathcal X}_u$ and a set of {\em safe states} ${\mathcal X}_s \triangleq {\mathcal L}_N \setminus {\mathcal X}_u$. Under what conditions is it possible to design  a  control input so that  every state trajectory that originates from a safe state remains  indefinitely in ${\mathcal X}_s$ and every trajectory that starts from an unsafe state 
 enters   ${\mathcal X}_s$
 in a finite number of steps and there remains?
\medskip

In \cite{EFMEV_JCD}  a complete characterization of the problem solvability  has been provided.
 \smallskip

\begin{proposition} \label{safecontrol} 
Given a BCN \eqref{bcnA} and the set ${\mathcal X}_u$ of unsafe states, the safe control problem
 is  solvable  if and only if %the following two conditions hold
 \begin{itemize}
 \item[i)] for every ${\bf x}\in {\mathcal X}_s = {\mathcal L}_N \setminus {\mathcal X}_u$ there exists ${\bf u}\in {\mathcal L}_M$ such that $L\ltimes {\bf u}\ltimes {\bf x}\in {\mathcal X}_s$;
 
 \item[ii)] the set ${\mathcal X}_s$ is reachable from every ${\bf x}\in {\mathcal X}_u$, which amounts to saying that
 for every ${\bf x}\in {\mathcal X}_u$ there exists $\bar {\bf x}\in {\mathcal X}_s$ such that $\bar {\bf x}$ is reachable from ${\bf x}$.
\end{itemize}
Moreover, if  the safe control problem is solvable, % for a   BCN (\ref{bcnA}),
then it is solvable by means of a state feedback law.
\end{proposition}
 \smallskip

{\color{black} \begin{remark} \label{differenze}
It is worth remarking that safe control to a  set ${\mathcal X}_s$ is similar to the problem of stabilizing a BCN to the set ${\mathcal X}_s$, but it does not coincide with it.  Indeed, while stabilization only requires that every state trajectory eventually enters ${\mathcal X}_s$, safe control imposes the additional constraint that states belonging to ${\mathcal X}_s$ can never leave ${\mathcal X}_s$, not even in a transient phase. 
\end{remark}
}
\smallskip

 The data-driven version of the safe control problem is the following one.
  \smallskip
 
 \begin{problem} \label{probl3}
Given the data $(\x_d$, $\ub_d)$ and the unsafe set ${\mathcal X}_u\subsetneq {\mathcal L}_N$, determine (if it exists) a state feedback control law $\ub(t) = K \x(t)$,  $K\in\mathcal L_{M\times N}$, that solves the safe control problem for  every BCN compatible with the data. 
%\begin{itemize}
%\item[1a)] Determine, if possible, the set $\mathcal X_e^d$ of the states in ${\mathcal L}_N$ that are equilibrium points of all the 
%  BCNs compatible with the data.
%\item[1b)] If $\mathcal X_e^d \ne \emptyset$ and $\x_e \in \mathcal X_e^d$, determine if there exists a state feedback matrix $K\in\mathcal L_{M\times N}$ such that $\ub(t) = K \x(t)$ stabilizes all the BCNs compatible with the data to the specific equilibrium point $\x_e$. 
%\end{itemize}
\end{problem}
 \medskip

When Problem \ref{probl3} is solvable, the data $(\x_d$, $\ub_d)$ will be called {\em informative for safe control} (with respect to the unsafe set ${\mathcal X}_u$). {\color{black} In the following theorem we provide necessary and sufficient conditions for the solvability of Problem \ref{probl3}.
\medskip

\begin{theorem} \label{sol3}
%Problem \ref{probl2} is solvable if and only if the data $(\x_d, \ub_d, \y_d)$ are informative for output regulation to $\y^*$ by state feedback. Moreover, for every $j \in [1,N]$, the $j$-th column of the feedback matrix $K\in \Lc^{M\times N}$ is given by 
%\be \label{Ko}
%\!\!\!\!{\rm col}_{r}(K) \!=\!
%\begin{cases}
% \ub_r \ \text{in} \ \eqref{ucycle}, \!\quad  \quad \quad \quad \forall r : \delta_N^r\in\mathcal C^d(\y^*), \\
% \ub_r \in \bigcup_{\bar \x \in \mathcal C^d(\y^*)}\ub_r^{\bar \x},   \forall r: \delta_N^r\notin\mathcal C^d(\y^*),
%\end{cases}
%\ee {\color{black} non riesco a capire la seconda espressione: me la spieghi?}{\color{red} Nel senso che se un determinato stato appartiene ai bacini di attrazione di più nodi in $C^d(\y^*)$, allora posso applicare come ingresso uno qualsiasi tra gli ingressi che mi restituisce l'algoritmo 2 applicato per i diversi nodi in $C^d(\y^*)$.} 
%where $\ub_r^{\bar \x}$ is the logical vector determined in Algorithm \ref{alg2}.
Consider the unsafe set ${\mathcal X}_u$ and set ${\mathcal X}_s = {\mathcal L}_N \setminus {\mathcal X}_u = \{ \delta^{i_1}_N, \delta^{i_2}_N, \dots, \delta^{i_d}_N\}$. Problem \ref{probl3} is solvable (i.e., the data  $(\x_d, \ub_d)$ are informative for safe control with respect to the unsafe set ${\mathcal X}_u$)
if and only  if 
 the following two conditions hold:
% \footnote{Note that since the set $\mathcal X_s$ must be   globally reachable  and each state in $\mathcal X_s$ must have at least one successor in ${\mathcal X}_s$, the matrix $X_p$ must be of full row rank. The reason is analogous to the one discussed in Remark \ref{frr}. }:

\begin{itemize}
\item[i)] for every $\delta_N^j\in {\mathcal X}_s$ there exist $k\in [1,T]$ and $\delta_N^\ell \in {\mathcal X}_s$ such that
\be
\begin{bmatrix} 
X_p \cr 
X_f 
\end{bmatrix} \delta_T^k = 
\begin{bmatrix}
\delta_N^{j} \cr 
\delta_N^\ell
\end{bmatrix}; 
\label{safe_succ}
\ee 
\item[ii)] 
  the row vector $(\sum_{\ell=1}^d \delta^{i_\ell}_N)^\top \left(\sum_{\tau=0}^{N-1} (L_{tot}^d)^\tau\right)$ has (all) positive entries.
  %, where
%$L_{tot}^d \triangleq X_f \odot_{\mathcal B} X_p^\top$. 
\end{itemize}
\end{theorem}

\begin{proof}
Problem \ref{probl3} is solvable if and only if conditions {\em i)} and {\em ii)} in Proposition \ref{safecontrol} are verified for all the BCNs compatible with the data.\\
In order to check condition {\em i)} of Proposition \ref{safecontrol},  we need to identify from data 
 if each state in ${\mathcal X}_s$ has (at least) one successor in ${\mathcal X}_s$.
%If we assume that ${\mathcal X}_s = \{\delta_N^{i_1},\delta_N^{i_2}, \dots, \delta_N^{i_q}\},$ with $1 \le i_1 < i_2 < \dots < i_q \le N$,
This amounts to verifying if for every $\delta^j_N\in {\mathcal X}_s$ there exists $k\in [1,T]$ and $\delta^\ell_N \in {\mathcal X}_s$ such that
\be
\begin{bmatrix} 
X_p \cr 
X_f 
\end{bmatrix} \delta_T^k = 
\begin{bmatrix}
\delta_N^{j} \cr 
\delta_N^\ell
\end{bmatrix} 
\label{safe_succ}
\ee 
which is exactly {\em i)} in the theorem statement.
If this is the case,   an input that allows the transition is
\be
\ub_j \triangleq U_p \delta_T^k.
\label{input_safetosafe}
\ee
On the other hand, by Proposition \ref{Ltot_data}, condition {\em ii)} in Proposition \ref{safecontrol} is trivially equivalent to condition {\em ii)} in 
the theorem.
\end{proof}

%  Moreover, 
%the $r$-th column of %for every $j \in [1,N]$, the $j$-th column of 
% the generic $r$-th column of a possible feedback matrix $K\in \Lc_{M\times N}$ is given by \ $
%{\rm col}_{r}(K) =
% \ub_r,  \forall r\in [1,N],
%where $\ub_r$ is determined through Algorithm \ref{alg2}.

Theorem \ref{sol3} establishes data-based conditions for the solvability of Problem \ref{probl3}; however, it does not yield an explicit solution, which is instead provided in Algorithm \ref{alg2}, below.
In 
  Algorithm \ref{alg2} 
  we first check condition {\em i)} of Theorem \ref{sol3}. If such condition is not verified, the algorithm stops and provides a negative outcome.
  If, on the other hand, the first check is successful, it performs the reachability check required by condition {\em ii)} from each state in ${\mathcal X}_u$ according to the same logic as in Algorithm \ref{alg1}.
  \\
  Algorithm \ref{alg2} returns, in Step 1, an input that keeps states already in $\mathcal X_s$ within 
$\mathcal X_s$, and in Step 2, an input that drives states in $\mathcal X_u$ closer to $\mathcal X_s$.
Such inputs are used to build a state feedback matrix $K$ that solves the problem.
 %  In the first step of Algorithm \ref{alg2}, for each state $\x \in \mathcal X_s$, an input $\ub^1$ is selected that keeps the state within $\mathcal X_s$. In the second step, for each state $\bar \x \in \mathcal X_u$, an input $\ub^2$ is chosen that drives the state closer to $\mathcal X_s$.
 % In the first step of the algorithm, for each state  ${\bf x}\in {\mathcal X}_s$ a possible choice of  ${\bf u}$  that 
 % ensures that the successor of ${\bf x}$ when ${\bf u}$ is applied belongs to ${\mathcal X}_s$ is given, meanwhile, in the second step, for each state in ${\mathcal X}_u$ a possible choice of the input values that drives the state closer to ${\mathcal X}_s$ is given.
 %checks and provides, if it exists, a logic matrix $K$ that solves the safe control problem for all BCNs compatible with the data.
\medskip

\begin{algorithm}[h]
\caption{Safe control} \label{alg2} 
\smallskip
{\color{black}
\textbf{Input:}  - The data matrices $X_p, X_f$ and $U_p$;\\
\hspace*{3.18em} - The set ${\mathcal X}_s= \{ \delta^{i_1}_N, \delta^{i_2}_N, \dots, \delta^{i_d}_N\}$. \\
\textbf{Output:} - Is  Problem 1 solvable?  Yes/No. \\
\hspace*{4em} - A state feedback matrix $K$ that solves the \\
\hspace*{4.85em}problem.  \smallskip \\
{\em Initialization:} $(\ub_1^1,\ub_2^1, \dots, \ub_N^1)= (\vect{0}, \vect{0},  \dots, \vect{0})$, \smallskip\\
\hspace*{6.25em}$(\ub_1^2,\ub_2^2, \dots, \ub_N^2)= (\vect{0}, \vect{0},  \dots, \vect{0})$. \\
\begin{enumerate}
    \item[\bf 1.] 
    \begin{itemize}
    \item Set $\mathcal S_0 = \emptyset$. \\
    \item {\tt for} $k \in [1,T]$, {\tt do}\\
\hspace*{1.5em}{\tt if}
\hspace*{4em}$$
\begin{bmatrix} 
X_p \cr 
X_f 
\end{bmatrix} \delta_T^k = 
\begin{bmatrix}
\delta_N^r \cr 
\delta_N^q
\end{bmatrix},
$$
\hspace*{2.5em} {\tt with} \ $\delta^r_N \in {\mathcal X}_s\setminus{\mathcal S}_0$\  {\tt and} \ $\delta^q_N\in {\mathcal X}_s$, {\tt then} \smallskip\\
\hspace*{4.5em}
${\mathcal S}_0 \leftarrow {\mathcal S}_0 \cup \{\delta^r_N\}$ and $\ub_r^1 \triangleq U_p \delta_T^k$ \\
\hspace*{1.8em}{\tt end if} \\
{\tt end for}\\ 
{\tt if} $\mathcal S_0 \not\equiv \mathcal X_s$, {\tt then} \\
\hspace*{2.5em} return (`No', {\tt `None'}). \\
{\tt end if}
\end{itemize}
    \item[\bf 2.] 
    \begin{itemize}
    \item Apply Algorithm \ref{alg1} with respect to the set $\mathcal X_s$. \\Let ($(\ub_1^2,\ub_2^2, \dots, \ub_N^2)$, $\mathcal S^*$) be the corresponding outcome. 
    \item {\tt for} $i \in [1,N]$, {\tt do} \\
    \hspace*{2.8em} $\ub_i = \ub_i^1 + \ub_i^2$ \\
    {\tt end for}
    \item Set $K = \begin{bmatrix}
        \ub_1 & \dots & \ub_N
    \end{bmatrix}$.
    \end{itemize}
\item[\bf 3.] {\tt if} $\mathcal S^* = \mathcal L_N$, {\tt then} \\
\hspace*{1.7em} return (`Yes', $K$) \\
 {\tt else}, \\
 \hspace*{2em}return (`No', {\tt `None'})\\
 {\tt end if} 
\end{enumerate}}
 \medskip
\end{algorithm}

}
Example \ref{ex1}, below, explores on a small size network  the results of this section 
pertaining data-driven   safe control.
 \smallskip
 
 \begin{example} \label{ex1}
 Consider a BCN \eqref{bcnA} with $N = 7$ and $M = 3$, described by the following matrix 
\begin{eqnarray*}
 L &=& \left[\begin{array}{c|c|c}
 L_1 & L_2 & L_3
 \end{array}\right]  = \big[  \delta_7^4 \ \delta_7^2 \ \delta_7^2 \ \delta_7^5 \ \delta_7^2 \ \delta_7^7 \ \delta_7^5 \ |\ \\
&&|\
\delta_7^6 \ \delta_7^1 \ \delta_7^3 \ \delta_7^2 \ \delta_7^4 \ \delta_7^5 \ \delta_7^7  \ |\  \delta_7^7 \ \delta_7^6 \ \delta_7^2 \ \delta_7^3 \ \delta_7^1 \ \delta_7^6 \ \delta_7^6\big]. 
 \end{eqnarray*}
%\begin{eqnarray*}
% L \!\!&=&\!\! \left[\begin{array}{c|c|c}
% L_1 & L_2 & L_3
% \end{array}\right] \\
% \!\!&=&\!\! \left[\begin{array}{c|c|}
%\delta_7^4 \ \delta_7^2 \ \delta_7^2 \ \delta_7^5 \ \delta_7^2 \ \delta_7^7 \ \delta_7^5 
% & 
%\delta_7^6 \ \delta_7^1 \ \delta_7^3 \ \delta_7^2 \ \delta_7^4 \ \delta_7^5 \ \delta_7^7 \end{array} \right. \\
%&&\left. \begin{array}{|c}
%\delta_7^7 \ \delta_7^6 \ \delta_7^2 \ \delta_7^3 \ \delta_7^1 \ \delta_7^6 \ \delta_7^6
% \end{array}\right]. 
% \end{eqnarray*}
  The equilibrium points of this BCN are $\{ \delta_7^2, \delta_7^3, \delta_7^6, \delta_7^7\}$. 
 We perform a single ($r=1$) offline experiment in the time interval $[0, T]$, with    $T= 12$,
  %\ge  N =7$, 
  and we collect the following data: 
 \begin{eqnarray*}
 X_p &=&\left[\begin{array}{c|c|c|c|c|c|c|c|c|c|c|c|c}
 \delta_7^1 &  \delta_7^7 &  \delta_7^7 &  \delta_7^6 &  \delta_7^5 &  \delta_7^1&  \delta_7^6 &  \delta_7^5 &  \delta_7^1 & \delta_7^4 & \delta_7^3 & \delta_7^3 & \delta_7^2
 \end{array}\right], \\
  U_p &=& \left[\begin{array}{c|c|c|c|c|c|c|c|c|c|c|c|c}
  \delta_3^3 &   \delta_3^2 &   \delta_3^3 &   \delta_3^2 &   \delta_3^3 &  \delta_3^2 & \delta_3^2 & \delta_3^3 & \delta_3^1 & \delta_3^3 & \delta_3^2 & \delta_3^1& \delta_3^1 
 \end{array}\right], \\
  X_f &=&\left[\begin{array}{c|c|c|c|c|c|c|c|c|c|c|c|c}
\delta_7^7 &  \delta_7^7 &  \delta_7^6 &  \delta_7^5 &  \delta_7^1 &  \delta_7^6 &  \delta_7^5 &  \delta_7^1 & \delta_7^4 & \delta_7^3 & \delta_7^3 & \delta_7^2 & \delta_7^2
 \end{array}\right]. 
 \end{eqnarray*}
 The BCNs compatible with these data can be generically represented by the following matrix $\tilde L$
\begin{eqnarray*}
 \tilde L &=& \left[\begin{array}{c|c|c}
 \tilde L_1 & \tilde L_2 & \tilde L_3
 \end{array}\right]  = \big[
\delta_7^4 \ \delta_7^2 \ \delta_7^2 \ * \ * \ * \ * 
 \ | \ \\ 
&&|\ \delta_7^6 \ * \ \delta_7^3 \ * \ * \ \delta_7^5 \ \delta_7^7 \ | \ \delta_7^7 \ * \ * \ \delta_7^3 \ \delta_7^1 \ * \ \delta_7^6
\big]. 
 \end{eqnarray*}
%\begin{eqnarray*}
% \tilde L \!\!&=&\!\! \left[\begin{array}{c|c|c}
% \tilde L_1 & \tilde L_2 & \tilde L_3
% \end{array}\right] \\
%\!\!&=&\!\! \left[\begin{array}{c|c|}
%\delta_7^4 \ \delta_7^2 \ \delta_7^2 \ * \ * \ * \ * 
% & 
%\delta_7^6 \ * \ \delta_7^3 \ * \ * \ \delta_7^5 \ \delta_7^7 \end{array} \right. \\
%\!\!&&\!\! \left. \begin{array}{|c}
%\delta_7^7 \ * \ * \ \delta_7^3 \ \delta_7^1 \ * \ \delta_7^6
% \end{array}\right]. 
% \end{eqnarray*}
 where the symbol $*$ stands for an arbitrary vector  (in ${\mathcal L}_7$). 
 {\color{black} Since 
 $$ L_{tot}^d = \begin{bmatrix}\delta^4_7+\delta^6_7+\delta^7_7 & \ \
 \delta^2_7 & \ \ \delta^2_7  + \delta^3_7 &\ \  \delta^3_7 &\ \  \delta^1_7 & \ \delta^5_7 &  \    \delta^6_7 +\delta^7_7 \end{bmatrix}$$}
 we obtain that the set of   the equilibrium points compatible with all the BCNs in $\mathcal B_d$ is 
 $
 \mathcal X_e^d = \{\delta_7^2, \delta_7^3, \delta_7^7\}. 
 $
 \begin{comment}
 We now check if there exists at least one state in $ \mathcal X_e^d$ that is globally reachable. To this end, we follow the procedure outlined in Algorithm \ref{alg2}. \\
 %\begin{itemize}
 $\bullet$  \ For $\x_e^1 \triangleq \delta_7^2$, we get: \\
 $\mathcal S_0 = \{\delta_7^2\}$, $\mathcal S_1 = \{\delta_7^3\}$, $\mathcal S_2 = \{\delta_7^4\}$, $\mathcal S_3 = \{\delta_7^1\}$, $\mathcal S_4 = \{\delta_7^5\}$, $\mathcal S_5 = \{\delta_7^6\}$, $\mathcal S_6 = \{\delta_7^7\}$, and hence $\mathcal S(\x_e^1) = {\mathcal L}_7$.\\
$\bullet$ \  For $\x_e^2 \triangleq \delta_7^3$, we get: \\
  $\mathcal S_0 = \{\delta_7^3\}$, $\mathcal S_1 = \{\delta_7^4\}$, $\mathcal S_2 = \{\delta_7^1\}$, $\mathcal S_3 = \{\delta_7^5\}$, $\mathcal S_4 = \{\delta_7^6\}$, $\mathcal S_5 = \{\delta_7^7\}$, $\mathcal S_6 = \emptyset$, and hence $\mathcal S(\x_e^2) \ne {\mathcal L}_7$. \\
 $\bullet$ \  For $\x_e^2 \triangleq \delta_7^7$, we get: \\
 $\mathcal S_0 = \{\delta_7^7\}$, $\mathcal S_1 = \{\delta_7^1\}$, $\mathcal S_2 = \{\delta_7^5\}$, $\mathcal S_3 = \{\delta_7^6\}$, $\mathcal S_4 = \emptyset$, and hence $\mathcal S(\x_e^2) \ne {\mathcal L}_5$. \\
% \end{itemize}
  Therefore, by Theorem \ref{sol1}, the data are informative for stabilization to $\x_e^1$ by state feedback. The state feedback matrix $K = \begin{bmatrix}
 \delta_3^1 &  \delta_3^1 &  \delta_3^1 &  \delta_3^3 &  \delta_3^3 & \delta_3^2 & \delta_3^3 
 \end{bmatrix}  \in \mathcal L_{3\times 7}
$ is obtained through Algorithm \ref{alg2}.
\smallskip
\end{comment}
\smallskip

Assume now that we want to implement a state feedback safe control with respect to the unsafe set $\mathcal X_u = \{\delta_7^3, \delta_7^4, \delta_7^7\}$. We need to understand if the data are informative for safe control with respect to $\mathcal X_u$, or, equivalently, if conditions i) and ii) of Theorem \ref{sol3} are satisfied. {\color{black} To this end, we can apply Algorithm \ref{alg2}, where in Step 1 we can readily verify that for every state in $\mathcal X_s = \mathcal L_7 \setminus \mathcal X_u = \{\delta_7^1, \delta_7^2, \delta_7^5, \delta_7^6\}$ there exists a transition ending in a state that is still inside $\mathcal X_s$, namely condition i) of Theorem \ref{sol3} holds. In Step 2, we instead examine condition ii) by employing Algorithm \ref{alg1} to determine the reachability of $\mathcal{X}_s$ in every BCN in $\mathcal{B}_d$. } We obtain: $\mathcal S_0 = \mathcal X_s$, $\mathcal S_1 = \{\delta_7^3, \delta_7^7\}$, and $\mathcal S_2 = \{\delta_7^4\}$, and hence $\mathcal S^* = \mathcal L_7$.  This means that also condition ii) of Theorem \ref{sol3} is satisfied, and thus the data are informative for safe control with respect to $\mathcal X_u$. Algorithm \ref{alg2} also provides the state feedback matrix $K = \begin{bmatrix}
 \delta_3^2 &  \delta_3^1 &  \delta_3^1 &  \delta_3^3 &  \delta_3^3 & \delta_3^2 & \delta_3^3 
 \end{bmatrix}$. %  \in \mathcal L_{3\times 7}$
%  is obtained via Algorithm 3. 

 \end{example} 
 \bigskip \medskip

\section{Data-driven output regulation %of a BCN 
%by state feedback}
} \label{p2}

{\color{black} We now address the second problem considered in this paper, namely output regulation, which consists in designing  a  control input such that the output of a BCN 
becomes constant and equal to 
a
desired value, after a finite number of steps. More specifically, we 
pursue this goal by designing 
a control action in  the form of a state feedback. 
 \smallskip

   Before addressing the problem in a data-driven framework, we review its definition and model-based solution (see, e.g., \cite{EF_MEV_BCN_obs2012,ACC2014_BCN}).} 
 \smallskip
 
 \begin{definition} \label{outreg}
Given a BCN \eqref{BCNtot}, the {\em output regulation problem to $\y^*\in {\mathcal L}_P$ by state feedback} is said to be solvable if  there exists $K \in \mathcal L_{M\times N}$  such that
for every $\x(0)\in \mathcal L_N$,  $\y(t) = \y^*, \forall t \ge \tau$,  $\exists \tau\in \mathbb Z_+$,   when $\ub(t) = K \x(t).$ %, t \in \mathbb Z_+$.}
\end{definition}
 \medskip

 The first step of the model-based solution consists in finding all the states that generate the desired output $\y^*$, namely the states belonging to the set $\mathcal X(\y^*) \triangleq \{\delta_N^j : H\delta_N^j = \y^*\}$.  Once we have determined this set, we have to ensure that,    {\color{black} under some state feedback law $\ub(t) = K\x(t), \exists K \in \mathcal L_{M\times N}$,  the state evolution of the resulting BN, described by 
\begin{equation}
  \label{eq.feedback}  
\x(t+1) =   L_K \x(t),
\end{equation}
where $L_K \triangleq L \ltimes K \ltimes \Phi_N \in \mathcal L_{N\times N},
$ and $\Phi_N$ is the {\em power-reducing matrix} \cite{BCNCheng}, satisfying $\Phi_N {\bf x}(t)= {\bf x}(t)\ltimes {\bf x}(t)$,}
is eventually constrained within $\mathcal X(\y^*)$. 
%from some time $\tau \in \mathbb Z_+$ onwards, 
This  imposes
%But, this is equivalent to determine the indistinguishability class in $1$ step corresponding to $\y^*$.
that {\color{black}
all limit cycles (in particular, equilibrium points) of the BN
\eqref{eq.feedback} have states
 inside $\mathcal X(\y^*)$.}\\
{\color{black} Since the effect of state feedback on a BCN is that of selecting, for each state, only one of its possible successors, corresponding to a specific choice of the input, each limit cycle appearing in the BN \eqref{eq.feedback} corresponds to a periodic state trajectory of the BCN \eqref{BCNtot}. On the other hand, the periodic state trajectories of the BCN \eqref{BCNtot} coincide with the cycles appearing in the digraph ${\mathcal D}(L_{tot})$. So, if
we denote by $\mathcal C(\y^*)$ the set of  
all cycles  in ${\mathcal D}(L_{tot})$ whose nodes correspond to states  belonging to $\mathcal X(\y^*)$, we first need to    verify that  $\mathcal C(\y^*) \ne \emptyset$, and then}  that it is globally reachable, meaning that the union of the  basins of attraction of the states in  $\mathcal C(\y^*)$ covers the whole state space $\mathcal L_{N}$. 
\smallskip 

In the following proposition, we formalize the necessary and sufficient conditions for solving the output regulation problem by state feedback in a model-based set-up. 
\smallskip

\begin{proposition} \cite{ACC2014_BCN} \label{regcns}
Given a BCN as in \eqref{BCNtot} and a desired  output value  $\y^*\in {\mathcal L}_P$, the problem of output regulation to $\y^*$ by state feedback is solvable
if and only if
%\item[ii)] 
%{\color{red} namely there exist $K\in \Lc^{M\times N}$ and $\tau\in \mathbb Z_+$ such that for every $\x(0)\in \Lc_N$ we obtain $\y(t) = \y^*, \forall t \ge \tau$, by applying $\ub(t) = K \x(t), t \in \mathbb Z_+$. }
\begin{itemize}
\item[i)] $\mathcal C(\y^*) \ne \emptyset$;
\item[ii)] $\mathcal C(\y^*)$ is globally reachable, i.e., $\cup_{\x \in \mathcal C(\y^*)}  {\mathcal S}(\x) \!=\! \Lc_N$.
\end{itemize}
\end{proposition}
\medskip

{\color{black} The data-driven version of the output regulation problem is given below.}
\smallskip

\begin{problem} \label{probl2}
Given the data $(\x_d, \ub_d, \y_d)$, and some desired value $\y^*\in \Lc_P$ for the output, determine (if possible) a state feedback matrix $K\in\mathcal L_{M\times N}$ such that, by applying $\ub(t) = K \x(t)$, for every initial condition $\x(0)$, the output trajectories of all the BCNs 
%in \eqref{bcnA}-\eqref{feedback}-\eqref{bcnAo} 
compatible with the data satisfy $\y(t) = \y^*$ for every $t\ge \tau, \exists \tau\in \mathbb Z_+$. 
\end{problem}
 \smallskip

When Problem \ref{probl2} is solvable, we will say that the data $(\x_d, \ub_d, \y_d)$ are {\em informative for output regulation to} $\y^*\in \Lc_P$ {\em by state feedback}. 

{\color{black} In order to solve Problem \ref{probl2},  we first determine from data} the set of  states that generate the desired output $\y^*$ in {\em all BCNs in} ${\mathcal B}_d$, i.e.,
%. We denote such a set by $\mathcal X^d(\y^*)$. It holds
%\be \label{Cstar}
$$\mathcal X^d(\y^*) \!=\! \left\{\delta_N^j \in \Lc_{N} : \exists k \in [1,T] \ \text{s.t.}  \begin{bmatrix}
X_p \cr 
Y_p 
\end{bmatrix}\delta_T^k =\! \begin{bmatrix}
\delta_N^j \cr
\y^*
\end{bmatrix}\right\}\!.  
$$ %\ee
To identify this set, we can proceed as in Algorithm \ref{alg3}. \medskip

\begin{algorithm}[h]
\caption{Construction of $\mathcal X^d(\y^*)$} \label{alg3} 
\smallskip
\textbf{Input:} - The data matrices $X_p$ and $Y_p$; \\
\hspace*{3.18em} - $\y^* \in \mathcal L_P$.\\
\textbf{Output:} The set $\mathcal X^d(\y^*)$ of all states  of the BCNs in \\
\hspace*{4.22em} ${\mathcal B}_d$ that generate the desired output $\y^*$. \smallskip \\
{\em Initialization:} Set $\mathcal X^d(\y^*) = \emptyset$. \\
\begin{enumerate}
    \item[\bf 1.] 
  {\color{black}  Compute 
    $$
    {\bf v}^\top = (\y^*)^\top Y_p \in \mathcal B^{1\times T},
    $$}
% $$
% \left[\begin{array}{c|c}
% I_N & \vect{0}_{N\times P} \cr 
% \hline
% \vect{0}_{1\times N} &  (\y^*)^\top
% \end{array}\right] \left[\begin{array}{c}
% X_p \cr 
% \hline
% Y_p
% \end{array}\right] = 
% \left[\begin{array}{c}
% X_p \cr 
% \hline
% {\bf v}^\top
% \end{array}\right]
% $$
and note that $[{\bf v}]_k \ne 0$ if and only if $X_p\delta_T^k \in \mathcal X^d(\y^*)$.
    \item[\bf 2.] 
    Define $\mathcal K \triangleq \{k_1, k_2, \dots, k_y\} =\  ${\tt nonzero}(${\bf v}$)
\item[\bf 3.] {\tt if} $\mathcal K \ne \emptyset$, {\tt then} \\
\hspace*{1.4em}{\tt for} $k \in \mathcal K$, {\tt do} \\
\hspace*{2.3em} {\tt if} $X_p\delta_T^k \notin \mathcal X^d(\y^*)$, {\tt then} \\
\hspace*{3.6em} $\mathcal X^d(\y^*) \leftarrow \mathcal X^d(\y^*) \cup \{X_p\delta_T^k\}$ \\
\hspace*{2.6em}{\tt end if} \\
\hspace*{1.4em}{\tt end for} \\
{\tt end if}
\item[\bf 4.] Return $\mathcal X^d(\y^*)$. 
\end{enumerate}
 \medskip
\end{algorithm}

 At this point, we need to verify if 
 {\color{black}there exist cycles whose nodes belong to}
 $\mathcal X^d(\y^*)$. 
 %To this end, let us define the data matrix  
%$$X \!\triangleq\!\begin{bmatrix}
%\x_d(0)  & \x_d(1)&\dots & \x_d(T_i-1) &\x_d(T_i)
%\end{bmatrix} \in \mathcal L_{N\times (T+1)},$$
%containing all the logical vectors of the sequence $\x_d$. 
The first step %of the following algorithm 
consists in determining all the transitions appearing in the data
%matrix $X$ 
involving (both initial and final) states that belong to $\mathcal X^d(\y^*)$. Based on these transitions, we then construct a digraph, whose set of nodes is  $\mathcal X^d(\y^*)$, and with edges corresponding to the transitions within $\mathcal X^d(\y^*)$ identified from data. We call {\color{black} $L_{tot}^d(\y^*)$} the 
  Boolean (but not necessarily logical) matrix associated to such a digraph.
  %, since we are accounting for transitions associated to different inputs. 
The matrix  $L_{tot}^d({\bf y}^*)$  represents the data-based estimate of  the principal submatrix of $L_{tot}$ in \eqref{ltot} obtained by selecting the rows and the  columns corresponding to 
the states in $\mathcal X^d(\y^*)$. 
%Indeed, as the data are  not informative for identifiability of the BCN, the   unitary entries of $L_{tot}^d({\bf y}^*)$ 
%in general are only a subset of  those of $L_{tot}$
%the transitions occurring among states in $\mathcal X^d(\y^*)$
 % for the original BCN (and for every BCN in ${\mathcal B}_d$). 
% in which we have discarded the transitions involving nodes that are not in $\mathcal X^d(\y^*)$. 
%In other words, it can be seen as the data-based reconstruction of  $L_{tot}$ in \eqref{ltot} in which we have 
%discarded the transitions involving nodes that are not in $\mathcal X^d(\y^*)$. 
%{\color{black} prima di tutto sottolineerei che in ogni caso stiamo tenendo conto di un sottoinsieme delle transizioni possibili dal momento che i dati non permettono la ricostruzione completa della BCN. In secondo luogo se ci restringiamo a un sottoinsieme dei nodi $[1,N]$ stiamo determinando non un'approssimazione di $L_{tot}$ ma di una sua sottomatrice principale che e' quella ottenuta selezionando sia le righe che le colonne indicizzate sugli indice dei vettori in $\mathcal X^d(\y^*)$. }
{\color{black} In Algorithm \ref{alg4}, we exploit the well-know Johnson's Algorithm (see \cite{JohnsonAlg}) to  identify the set of all cycles  in ${\mathcal D}(L_{tot}^d({\bf y}^*))$, say $\mathcal C^d(\y^*)$.}
% When working with data, it is convenient to separately identify limit cycles of unitary length, i.e., equilibrium points, and those of length greater than one. 
% To determine the set of equilibrium points (if any) contained in $\mathcal X^d(\y^*)$, say $\mathcal X_e^d(\y^*)$, we can exploit Algorithm \ref{alg1}. Hence, we get $\mathcal X_e^d(\y^*) \triangleq \mathcal X_e^d \cap \mathcal X^d(\y^*)$. Instead, to identify potential limit cycles contained in $\mathcal X^d(\y^*)$ from the collected data, we can apply the following algorithm, which 
%The algorithm is based on the idea that a limit cycle is compatible with the collected data if and only if it is possible to recognize some cyclic pattern on the data matrix $X \triangleq \begin{bmatrix}
%\x_d(0) & \x_d(1) & \dots & \x_d(T-1) & \x_d(T)
%\end{bmatrix} \in \mathcal L_{N\times (T+1)}$. 
%{\color{red} La questione della possibilità di effettuare più esperimenti la lasciamo vaga? Perché altrimenti qui bisognerebbe anche specificare che si deve prendere in considerazione ogni singolo esperimento in maniera separata.}  {\color{black}metterei una footnote dicendo che l'algoritmo prende in esame una singola serie storica. Se ce ne sono altre va ripetuto per ciascuna serie e i risultati assemblati. Omesso due to page constraint ;-)}
\medskip

\begin{algorithm}[h]
\caption{Cycles in $\mathcal D(L_{tot}^d(\y^*))$}
%in $\mathcal X^d(\y^*)$} 
\label{alg4} 
\smallskip
\textbf{Input:} - The data matrix $X$; \\
\hspace*{3.18em} - The set $\mathcal X^d(\y^*) = \{ \delta^{i_1}_N, \dots, \delta^{i_s}_N\}$.\\
\textbf{Output:} The set $\mathcal C^d(\y^*)$ of all   cycles in $\mathcal D(L_{tot}^d(\y^*))$. \smallskip \\
{\em Initialization:} Set:
\begin{itemize}
\item $P \triangleq \begin{bmatrix}
\delta_N^{i_1} & \delta_N^{i_2} & \dots & \delta_N^{i_s}
\end{bmatrix} \in \Lc_{N\times s}
$;
%where $s$ is the cardinality of $\mathcal X^d(\y^*)$ 
\item $\tilde{\mathcal V} \triangleq \{i_1,i_2,\dots,i_s\} \subset [1,N]$;
%is the set of indices of the states in $\mathcal X^d(\y^*)$ (taken in an arbitrary order); 
\item $\mathcal C^d(\y^*) = \emptyset$;
\item $\mathcal D(L_{tot}^d({\bf y}^*)^i) \triangleq (\tilde{\mathcal V}, \tilde{\mathcal E})$,  with 
%$\tilde {\mathcal V} = \mathcal I$ and 
$\tilde{\mathcal E} = \emptyset$.
\end{itemize}
\smallskip 

\begin{enumerate}
    \item[\bf 1.] {\color{black} {\tt for} $(i,j) \in [1,s]\times [1,s]$, {\tt do} \\
\hspace*{1em} {\tt if} $\exists k \in [1, T]$ s.t. $[P^\top X_p]_{ik}[P^\top X_f]_{jk} \ne 0$, {\tt then} \\
\hspace*{2.8em} $\tilde{\mathcal E} \leftarrow \tilde{\mathcal E} \cup \{(i,j)\}$ \\
\hspace*{1em} {\tt end if} \\
{\tt end for}} 
 \item[\bf 2.] Apply Johnson's Algorithm in \cite{JohnsonAlg} to find all cycles in $\mathcal D(L_{tot}^d({\bf y}^*))$, say $\tilde{\mathcal C}_1, \dots, \tilde{\mathcal C}_{n_c}$,   $n_c \in \mathbb Z_+$. 
 \item[\bf 3.] Remap all the nodes of $\mathcal D(L_{tot}^d({\bf y}^*))$ to the corresponding nodes in $[1,N]$, by using the bijective correspondence given by the matrix $P$, i.e., $1\leftrightarrow i_1, \dots, s \leftrightarrow i_s$. 
\item[\bf 4.] Define $\mathcal C_1, \dots, \mathcal C_{n_c}$ the remapped   limit cycles.
\item[\bf 5.] Return $\mathcal C^d(\y^*)$.
\end{enumerate}

\end{algorithm}

For each edge $(i,j)$ belonging to some cycle  ${\mathcal C}_\ell$ in $\mathcal C^d(\y^*)$
%, say $\mathcal C_l$, 
it is also possible to determine one of the inputs associated with such transition within the cycle, say $\ub_i^\ell$. Indeed, if there exists an edge $(i,j)$, it means that the transition $i \to j$ is captured by the collected data, or, equivalently, there  exists $k \in [1,T]$ such that 
\be \label{lc_data}
\begin{bmatrix} 
X_p \cr 
X_f
\end{bmatrix} \delta_T^k = 
\begin{bmatrix}
\delta_N^i \cr 
\delta_N^j
\end{bmatrix}.
\ee 
Therefore, %we also know
  the input associated to this transition within ${\mathcal C}_\ell$
is 
\be \label{lc_input}
\ub_i^\ell \triangleq U_p \delta_T^k.
\ee
%
%Otherwise, let $d$ be the dimension of one such of submatrices. Obviously, $d \le s$.  Let $k_1, k_2, \dots, k_d$ be the indices of the columns of $X_{\mathcal X^d(\y^*)}$ corresponding to the cyclic submatrix. The sequence $k_1\to k_2\to \dots k_s$ identifies a limit cycle, say $\mathcal C_1$. The inputs related to the transitions within the limit cycle are $U_p\delta_M^{k_1}, U_p\delta_M^{k_2}, \dots, U_p\delta_M^{k_d}$. \\
%Repeat the same procedure for each $i$-th cyclic submatrix with $i \in [2,n_c]$, obtaining $\mathcal C_i$. \\
%Return $\mathcal C^d(\y^*) = \{\mathcal C_1, \dots, \mathcal C_2\}$. 
%
%%%%%%%%%%%
%The $k$-th column of $X_{\mathcal X^d(\y^*)}$ is nonzero both in the upper and in the lower part if and only if the corresponding columns of $X_p$ and $X_f$ are both vectors in $\mathcal X^d(\y^*)$. 
%Check whether $X_{\mathcal X^d(\y^*)}$ contains some cyclic submatrices. Let $n_c$ be the number of such submatrices.
%, return $\mathcal C^d(\y^*) = \emptyset$. \\
%%%%%%%%%%%
 At this point   it only remains to establish if the set $\mathcal C^d(\y^*)$ is globally reachable.
 {\color{black} This can be verified using Algorithm \ref{alg1}, which additionally specifies the inputs to be applied to the states outside $\mathcal C^d(\y^*)$. } 
 %(a slightly amended version of) Algorithm 3. 
 \smallskip

We can  now provide the solution to Problem \ref{probl2}. 
 \smallskip

%\begin{definition} \label{info_reg_def}
%{\color{red} %Given the BCN \eqref{BCNtot}, 
%We say that the data $(\x_d, \ub_d, \y_d)$ are {\em informative for output regulation to} $\y^*\in \Lc_P$ {\em by state feedback}}  if the output evolution of all the BCNs in $\mathcal B_d$ can be stabilized to $\y^*$ (i.e., $\exists \tau \in \mathbb Z_+$ s.t. $\y(t) =\y^*, \forall t\ge \tau$) by applying  {\color{red} the feedback control law} $\ub(t) = K\x(t), \exists K \in \mathcal L_{M\times N}$.  
%\end{definition}

%By point {\em ii)} of Proposition \ref{regcns}, informativity for output regulation by state feedback can be characterized as follows. 

\begin{theorem} \label{sol2}
%Problem \ref{probl2} is solvable if and only if the data $(\x_d, \ub_d, \y_d)$ are informative for output regulation to $\y^*$ by state feedback. Moreover, for every $j \in [1,N]$, the $j$-th column of the feedback matrix $K\in \Lc^{M\times N}$ is given by 
%\be \label{Ko}
%\!\!\!\!{\rm col}_{r}(K) \!=\!
%\begin{cases}
% \ub_r \ \text{in} \ \eqref{ucycle}, \!\quad  \quad \quad \quad \forall r : \delta_N^r\in\mathcal C^d(\y^*), \\
% \ub_r \in \bigcup_{\bar \x \in \mathcal C^d(\y^*)}\ub_r^{\bar \x},   \forall r: \delta_N^r\notin\mathcal C^d(\y^*),
%\end{cases}
%\ee {\color{black} non riesco a capire la seconda espressione: me la spieghi?}{\color{red} Nel senso che se un determinato stato appartiene ai bacini di attrazione di più nodi in $C^d(\y^*)$, allora posso applicare come ingresso uno qualsiasi tra gli ingressi che mi restituisce l'algoritmo 2 applicato per i diversi nodi in $C^d(\y^*)$.} 
%where $\ub_r^{\bar \x}$ is the logical vector determined in Algorithm \ref{alg2}.
Problem \ref{probl2} is solvable (i.e., the data $(\x_d, \ub_d, \y_d)$ are informative for output regulation to $\y^*\in {\mathcal L}_P$ by state feedback)
if and only  if 
 the following two conditions hold:
% \footnote{Note the fact that the set $\mathcal C^d(\y^*)$ must be nonempty and globally reachable automatically implies that $X_p$ must be of full row rank. The reason is analogous to the one discussed in Remark \ref{frr}. But if this is the case, the matrix $H$ in \eqref{bcnAo} is uniquely determined from the collected data.}:

\begin{itemize}
\item[i)] $\mathcal C^d(\y^*) \ne \emptyset$;
\item[ii)] $ { \mathcal S^*} \triangleq \bigcup_{\x \in \mathcal C^d(\y^*)} \mathcal S(\x) = \mathcal L_N$. %$\exists \tilde \x \in \mathcal C^d(\y^*)$ such that the data $(\x_d, \ub_d)$ are informative for reachability of $\tilde \x$
\end{itemize}
\end{theorem}

\begin{proof}
(Necessity). Assume, by contradiction, that condition {\em i)} does not hold, i.e.,  $\mathcal C^d(\y^*) = \emptyset$. 
But this means that there exists at least one BCN in $\mathcal B_d$ for which $\mathcal C(\y^*) = \emptyset$,  and hence,   by Proposition \ref{regcns}, for that BCN the
 output regulation to $\y^*$ by state feedback is not possible.

  Assume now that $\mathcal C^d(\y^*) \ne \emptyset$, but {\em ii)} does not hold, namely $\mathcal S^* \ne \mathcal L_N$. 
We can construct a BCN in ${\mathcal B}_d$ for which 
 $\mathcal C(\y^*) = \mathcal C^d(\y^*)$ and 
$\bigcup_{\x \in \mathcal C(\y^*)} \mathcal S(\x) =  { \mathcal S^*} \ne \mathcal L_N$.
 This requires to select all the state transitions, equivalently  the columns of $\tilde L$, that are not uniquely identified by the data (and possibly also the corresponding output).
If there exists $\delta^r_N$ that does not appear as a column of $X_p$, then for every $i
\in [1,M]$ we impose $\tilde L_i \delta^r_N =  \delta^r_N$  and $H\delta^r_N \ne {\bf y}^*$.
On the other hand,  if
$\delta^r_N = X_p\delta^k_T$ for some $k\in [1,T]$, we define
 $\mathcal I_r$ the set of indices $i\in [1,M]$ such that $U_p\delta^k_T = \delta^i_M$.
If $\mathcal I_r \subsetneq [1,M]$, we impose    $\tilde L_i \delta^r_N =  \delta^r_N$ for every $i\not\in \mathcal I_r$.
Clearly, by Proposition \ref{regcns},  for each BCN obtained in this way %$\mathcal C(\y^*) = \mathcal C^d(\y^*)$ and ${\mathcal S^*} \ne \mathcal L_N$.
%Consequently, 
the output regulation to $\y^*$ by state feedback is not possible. \\

(Sufficiency). It is a direct consequence of Proposition \ref{regcns} and following comments.  %The final part follows trivially from Algorithm 3.
%
%Hence, the input $\ub_r = \ub_r^{\bar \x}$ with $\ub_r^{\bar \x}$ (determined in Algorithm \ref{alg2}) is the input (or one of the inputs) that, starting from the state $\x$, ensures the correct transition towards reaching $\bar \x$, by construction. If $\x$ belongs to $\mathcal S(\bar \x)$ for more than one $\bar \x \in \mathcal C^d(\y^*)$, the previous reasoning can be repeated for any of such $\bar \x$'s, and hence we can select any of the corresponding inputs. 
\end{proof}
\medskip

{\color{black} As for Theorem \ref{sol3} in Section \ref{sec.SC}, Theorem \ref{sol2} only provides a method to verify, using data, whether the output regulation problem is solvable. However, it is not constructive, in the sense that it does not offer a systematic procedure for explicitly computing the solution, namely the feedback matrix $K \in \mathcal L_{M\times N}$. Nevertheless, the matrix $K$ can be readily obtained 
as
$$K = \begin{bmatrix}\ub_1 & \dots & \ub_N\end{bmatrix}$$
by selecting as control input $\ub_i$ the one defined in \eqref{lc_input} for the states $\delta^i_N$ belonging to a   cycle in $\mathcal C^d(\y^*)$, and using the input $\ub_i$ provided by Algorithm \ref{alg1} for all other states $\delta^i_N$. In cases where a state belongs to multiple  cycles in $\mathcal C^d(\y^*)$, it suffices to select the control input associated with any of these  cycles.     
\medskip
%the $r$-th column of %for every $j \in [1,N]$, the $j$-th column of 
%  the generic $r$-th column of a possible feedback matrix $K\in \Lc_{M\times N}$ is given by $
% {\rm col}_{r}(K) =
%  \ub_r,  \forall r\in [1,N],
% $
% where $\ub_r$ is determined by Algorithm 3.

\begin{remark} \label{of_comp}
In this section, we have explored how to solve the output regulation problem via state feedback based on collected data. Clearly, this strategy cannot be implemented unless offline data about the state have been collected.
An alternative approach, first explored in \cite{BCNdd}, is to perform output regulation by means of output feedback. 
When so, one could try to achieve this goal by collecting only input and output data.
The solution proposed in \cite{BCNdd}, and
  inspired by results on the stabilization of Probabilistic Boolean networks, relies on an algorithm that first computes an approximation of the output prediction matrix, based on data, and then generates a sequence of concentric annular sets, each containing the output values that lie at a specific ``distance" from the desired value $\y^*$, with respect to the output/input/output transitions captured by the data-based output prediction matrix.  Finally, using these sets, the algorithm computes, when possible, an output feedback gain.  

However, output feedback strategies, which are demanding even in model-based settings, become hardly applicable in a data-driven framework, where only part of the network transitions is accessible.
\end{remark}
}
\medskip

\begin{example} \label{ex2}
Consider a BCN \eqref{BCNtot} with $N = 6$, $M = 3$, and $P = 2$, described by the following matrices 
\begin{eqnarray*}
L &=& \left[\begin{array}{c|c|c}
 L_1 & L_2 & L_3
 \end{array}\right] =  \left[\begin{array}{c|}
\delta_6^2 \ \delta_6^4 \ \delta_6^3 \ \delta_6^3 \ \delta_6^6   \ \delta_6^5 \end{array}\right. \\
&&\left.\begin{array}{|c|c} 
\delta_6^1 \ \delta_6^5 \ \delta_6^2 \ \delta_6^2 \ \delta_6^6  \ \delta_6^1 
 &
\delta_6^5 \ \delta_6^1 \ \delta_6^4 \ \delta_6^5 \ \delta_6^4  \ \delta_6^6  
 \end{array}\right], \\
 H &=& \left[\begin{array}{c}
 \delta_2^1 \ \delta_2^2 \ \delta_2^2 \ \delta_2^2 \ \delta_2^1 \ \delta_2^1
 \end{array}\right]. 
 \end{eqnarray*}
 Assume that we want to regulate the output of the network to the value $\y^* = \delta_2^2$. 
 %The limit cycles (or in particular equilibrium points) of the BCN are $\{ \delta_6^1, \delta_6^3, \delta_6^6, \delta_6^1 \to \delta_6^2 \to \delta_6^1\}$. 
 We perform a single ($r=1$) offline experiment in the time interval $[0, 9]$, % with $T=9 \ge N = 6$, 
 and   collect the  data: 
 \begin{eqnarray*}
 X &=& \left[\begin{array}{c|c|c|c|c|c|c|c|c|c}
 \delta_6^6 &  \delta_6^6 &  \delta_6^1 &  \delta_6^2 &  \delta_6^5 &  \delta_6^4 & \delta_6^2 &  \delta_6^4 &  \delta_6^3 &  \delta_6^3
 \end{array}\right], \\
  U_p &=& \left[\begin{array}{c|c|c|c|c|c|c|c|c}
  \delta_3^3 &   \delta_3^2 &   \delta_3^1 &   \delta_3^2 &   \delta_3^3 &   \delta_3^2 &  \delta_3^1 &  \delta_3^1 &  \delta_3^1 
 \end{array}\right],  \\
Y_p &=& \left[\begin{array}{c|c|c|c|c|c|c|c|c}
\delta_2^1 & \delta_2^1 & \delta_2^1 & \delta_2^2 & \delta_2^1 & \delta_2^2 & \delta_2^2 & \delta_2^2 & \delta_2^2
 \end{array}\right]. 
 \end{eqnarray*}
 %The set $\mathcal X^d(\y^*)$ is given by 
 We deduce $
 \mathcal X^d(\y^*) = \{\delta_6^2, \delta_6^3, \delta_6^4\}. 
 $
 The BCNs compatible with these data can be  represented by the following matrix %$\tilde L$
\begin{eqnarray*}
\tilde L &=& \left[\begin{array}{c|c|c}
 \tilde L_1 & \tilde L_2 & \tilde L_3
 \end{array}\right] =  \left[\begin{array}{c|}
\delta_6^2 \ \delta_6^4 \ \delta_6^3 \ \delta_6^3 \ *   \ * \end{array}\right. \\
&&\left.\begin{array}{|c|c} 
* \ \delta_6^5 \ * \ \delta_6^2 \ *  \ \delta_6^1 
 &
* \ * \ * \ * \ \delta_6^4  \ \delta_6^6  
 \end{array}\right],
 \end{eqnarray*}
 where, again, the symbol $*$ stands for an arbitrary vector (in $\Lc_6$). As  %already noticed,  being 
  $X_p$ is of full row rank, the matrix $H$ can be uniquely recovered. Hence, by applying Algorithm \ref{alg4}, we obtain that the set of all     cycles   in $\mathcal D(L_{tot}^d(\y^*))$ %$\mathcal X^d(\y^*)$
 %compatible with the BCNs in $\mathcal B_d$ 
 is \ $
 \mathcal C^d(\y^*) = \{
 3, 
 2 \leftrightarrow 
4\},
 $
  where   $ 2 \leftrightarrow 4$   is the  cycle of length $2$ involving $\delta_6^2$ and $\delta_6^4$. 
 We now check if the set  $\mathcal C^d(\y^*)$ is globally reachable.  By following the procedure outlined in Algorithm \ref{alg1}, we obtain $\mathcal S^* = \Lc_6$. 
% \begin{itemize}
% \item $\mathcal S(\delta_6^2) = \{\delta_6^1, \delta_6^2, \delta_6^4, \delta_6^6\}$;
% \item $\mathcal S(\delta_6^3) = \mathcal L_6$; 
%  \item $\mathcal S(\delta_6^4) = \{\delta_6^1, \delta_6^2, \delta_6^4, \delta_6^5, \delta_6^6\}$. 
% \end{itemize}
 Thus, by Theorem \ref{sol2}, the data are informative for output regulation to $\y^*$ by state feedback. {\color{black} Building on the argument presented after Theorem \ref{sol2}, we obtain the following feedback matrix}
$$   
 K = \begin{bmatrix}
 \delta_3^1 &  \delta_3^1 &  \delta_3^1 &  \delta_3^2 &  \delta_3^3 & \delta_3^2 
 \end{bmatrix}.
 $$
  %via Algorithm 3. % and 5.} % (see also Theorem \ref{sol2}), 
% and a  possible choice is
 \end{example}

\section{Conclusion} \label{concl}
In this paper we investigated how the recently introduced data informativity approach \cite{informativity} can be adapted and extended to solve fundamental control problems in Boolean control networks. More specifically, we focused on safe control and output regulation problems, and we provided necessary and sufficient conditions, based only on  previously collected data (supposed not to be informative for identification), for their solvability. Moreover, we also obtained from data the feedback matrices expressions. We validated our results by means of examples.
{\color{black} It is worth noting that, due to the inherent complexity of these systems, BCN problems are NP-hard with respect to the state, input and output dimensions $n$, $m$ and $p$, or equivalently polynomial in the corresponding quantities $N = 2^n$, $M = 2^m$ and $P = 2^p$. Importantly, the computational complexity is independent of the chosen representation: in the case of interest, the algebraic one. Moreover, relying on data instead of an exact model does not substantially change the overall computational complexity. The only difference is that it also depends on the collected data, and in particular on the number of collected data $T$, which is, however, lower-bounded by $N$. On the other hand, since data provide only partial knowledge of the system dynamics, the number of identified transitions is smaller than those appearing in the real model, and consequently, if their complete enumeration is required, it imposes a lower computational burden.  Nonetheless, depending on the problem at hand, the algorithms can be easily adapted to terminate once a {\em single} solution is found, without the need to search for {\em all} feasible solutions.  }

  \bibliographystyle{agsm}       % Include this if you use bibtex 
%\bibliography{Refer188}           % and a bib file to produce the 

\end{document}